\documentclass[11pt]{article}
\usepackage{showlabels}
\usepackage{epsfig,psfrag,amsmath,amssymb,latexsym}
\usepackage{amscd }
\usepackage{amsfonts}
\usepackage{graphicx}
\usepackage[T1]{fontenc}
\usepackage{xcolor}
\usepackage[colorlinks=true]{hyperref}
\usepackage{amsthm, amscd,epsfig,enumerate, xcolor}
\newtheorem{lemma}{Lemma}[section]

\newtheorem{proposition}{Proposition}[section]

\numberwithin{equation}{section}

\def\Qt{q_{\theta}}
\def\Qy{q_{Y}}
\def\hQt{\hat{q}_{\theta}}
\def\hQy{\hat{q}_{Y}}

\hypersetup{colorlinks,citecolor=blue,linkcolor=blue,urlcolor=blue}
\def\university#1{{\sl \begin{center} #1 \vspace{2pt} \end{center} } }
\def\inst#1{\vspace{1pt} \unskip$^{#1}$}
\begin{document}
\title{Estimation of Viterbi path in Bayesian hidden Markov models}
\author{J\"{u}ri Lember\inst{1}\footnote{Corresponding author, e-mail: jyri.lember@ut.ee}, Dario Gasbarra\inst{2},  Alexey Koloydenko\inst{3} and Kristi Kuljus\inst{1}}
\maketitle
\university{\inst{1}University of Tartu, Estonia; \inst{2}University of Helsinki, Finland;\\ \inst{3}Royal Holloway, University of London, UK}
\begin{abstract}
The article studies different methods for estimating the Viterbi
path in the Bayesian framework. The Viterbi path is an estimate of
the underlying state path in hidden Markov models (HMMs), which has
a maximum joint posterior probability. Hence it is also called the
maximum a posteriori (MAP) path. For an HMM with given parameters,
the Viterbi path can be easily found with the Viterbi algorithm. In
the Bayesian framework the Viterbi algorithm is not applicable and
several iterative methods can be used instead. We introduce a new
EM-type algorithm for finding the MAP path and compare it with
various other methods for finding the MAP path, including the
variational Bayes approach and MCMC methods. Examples with simulated
data are used to compare the performance of the methods. The main
focus is on non-stochastic iterative methods and our results show
that the best of those methods work as well or better than the best
MCMC methods. Our results demonstrate that when the primary goal is
segmentation, then it is more reasonable to perform segmentation
directly by considering the transition and emission parameters as
nuisance parameters.
\end{abstract}

\section{Introduction and preliminaries}
Hidden Markov models (HMMs) are widely used in application areas
including speech recognition, computational linguistics,
computational molecular biology, and many more. Recently there has
been a continuing interest to apply HMMs in the Bayesian framework,
where model parameters are assumed to have a prior distribution. The
Bayesian approach has the advantage that it allows researchers to
incorporate their prior beliefs and information in the modeling
process. However the Bayesian framework might heavily complicate the
analysis, since a mixture of HMMs is not typically an HMM any more.
Therefore the algorithms and methods valid for a single HMM might
not be applicable in the Bayesian setup. For example, when several
HMMs are mixed, the optimality principle no longer holds and dynamic
programming algorithms such as the Viterbi algorithm and
forward-backward algorithms do not work. Therefore finding the
Viterbi path in the Bayesian framework is a difficult task where no
simple solution exists.

Most of the literature on HMMs in the Bayesian framework (see, e.g.
\cite{HMMbook,koski,RobertMarin,Scott,Ryden,titterington,corander})
deals with MCMC methods \cite{Scott,HMMbook,koski,Ryden,Besag}. When
the goal is estimation of the underlying hidden state path for a
given observation sequence which we refer to as {segmentation} (also
terms {decoding, denoising} are used), various methods based on
Gibbs sampling, for example simulated annealing, are often used.
Note that simulated annealing works well only if the applied cooling
schedule is correct and the number of sweeps large enough. In this
article, we study and compare different non-stochastic methods for
finding the Viterbi path in the Bayesian framework, because in
comparison to MCMC methods non-stochastic methods are
computationally less demanding. We introduce a new EM-type
segmentation method which we call {\it segmentation EM}, and give an
overview of other most commonly used non-stochastic segmentation
methods. Unlike in the traditional EM algorithm, in the segmentation
EM algorithm the hidden path is considered as the main parameter of
interest. The performance of segmentation EM is compared with the
other segmentation methods (including the variational Bayes approach
and parameter EM estimation procedures) in numerical examples.
According to our numerical examples, the segmentation EM method and
a closely related method which we call {\it segmentation MM},
perform at least as well or better than MCMC methods. Moreover our
empirical studies demonstrate that the direct Bayesian segmentation
approach outperforms the commonly used parameters-first approach,
where segmentation is performed after the parameters have been
estimated.

Viterbi path estimation in the Bayesian framework has been studied and applied in speech tagging problems \cite{tag,johnson,garca,johnson2}.
An overview of Bayesian HMMs in speech tagging can be found in \cite{two}.
These papers study several methods for calculating the Viterbi path, including simulated annealing, variational Bayes, and also the parameters-first approach.
The results of the studies are contradictory (see, e.g. \cite{johnson2}), showing that further research in this area is needed.

This article is organized as follows. In the rest of this section we introduce the problem of estimating the Viterbi path in the frequentist and Bayesian frameworks;
at the end of the section the main objectives of the article will be summarized.
Section \ref{methods} gives a brief overview of the methods and algorithms that we consider.  Section \ref{simulations} presents the results of numerical examples: first, in subsection
\ref{Imodel}, we consider the case where the emission parameters are
known and transition parameters are unknown having Dirichlet
prior distributions; then in subsection \ref{caseII} we consider the case where the
emission parameters are also unknown. In Section \ref{HyperRole} the role of hyperparameters and their effect on segmentation results is discussed.
Section \ref{disc} explains the relationships and similarities of the  segmentation algorithms that we study.
The formulae needed to apply the  segmentation methods are presented in the Appendix.

\subsection{Segmentation with hidden Markov models}
\paragraph{Hidden Markov model.}
\noindent Consider a homogeneous Markov chain $Y^n:=Y_1,\ldots,Y_n$
with states $S=\{1,\ldots,K\}$. Let $X^n:=X_1,\ldots,X_n$ be random
variables taking values on ${\cal X}$ such that: 1) given $Y^n$, the
random variables $\{X_t\}$, $t=1,\ldots,n$, are conditionally
independent; 2) the distribution of $X_t$ depends on $Y^n$ only
through $Y_t$. Since only $X^n$ is observed, the pair $(Y^n,X^n)$ is
referred to as a {\it hidden Markov model}. Because the underlying
Markov chain is homogeneous, the model is fully specified by the
transition matrix $\mathbb{P}=(p_{lj})$, $l,j=1,\ldots,K$, initial
probabilities $p_{0k}$, $k=1,\ldots,K$, and {\it emission
distributions} $P(X_t\in \cdot |Y_t=k)$, $k=1,\ldots,K$. Thus there
are two types of parameters in our model: transition parameters
$\Theta_{tr}$ where $\theta_{tr}\in \Theta_{tr}$ specifies the
transition matrix and the initial probability vector, and emission
parameters $\Theta_{em}$. Often the initial distribution is fixed or
a function of the transition matrix. In this case any $\theta_{tr}$
can be identified with a transition matrix. The whole parameter
space is given by $\Theta:=\Theta_{em}\times \Theta_{tr}$. Without
loss of generality we assume that all emission distributions have
{\it emission densities} $f_k$ with respect to some common reference
measure on ${\cal X}$. Typically, all emission densities are assumed
to belong to the same parametric family ${\cal F}=\{f(\cdot|\theta):
\theta\in \Theta_{em}\}$. Thus, for any state $k\in S$, there is a
$\theta^k_{em}\in \Theta_{em}$ such that
$f_k(\cdot)=f_k(\cdot|\theta^k_{em})$. For any realization $x^n$ of
the random variables $X^n$ and for any realization $y^n$ of the
Markov chain $Y^n$, let the joint likelihood of $(y^n,x^n)$ be
denoted by $p(y^n,x^n)$. Similarly, $p(x^n)$ and $p(y^n)$ denote the
marginal likelihoods of $p(y^n,x^n)$, and $p(x^n|y^n)$ and
$p(y^n|x^n)$ stand for the conditional likelihoods.
We assume that the length of the observation sequence is fixed and leave it from the notation. Thus, we
 denote by $x\in {\cal X}^n$ a vector $x^n$ of observations
and by $y,s\in S^n$ state sequences. Also, $X$ and $Y$ stand for
$X^n$ and $Y^n$, respectively, $p(y,x)$ is used instead
of $p(y^n,x^n)$ and so on. To indicate a
single entry of a vector $x$ or $s$, we use $x_t$ or $s_t$, $t= 1,\ldots,n$. For any $t=1,\ldots,n$, $p_t$ is used
for marginal probability, for example
$$p_t(y_t|x)=P(Y_t=y_t|X=x)=\sum_{s: s_t=y_t}p(s|x).$$
\paragraph{Viterbi path.} Suppose $\theta=(\theta_{tr},\theta_{em}) \in \Theta$ is given, that is both the transition and emission parameters are given. The {\it segmentation problem} consists of
estimating the unobserved realization of the underlying Markov chain
$Y$  given  observations $X$. Formally, we are looking for a
mapping $g:{\cal X}^n \to S^n$ called a {\it classifier}, that maps
every sequence of observations into a state sequence. The best
classifier $g$ is often defined via a {\it loss function}, for an overview of risk-based segmentation with HMMs based on loss functions, see \cite{seg,intech,chris}.
The {\it Viterbi path} for given parameters $\theta$ is defined as a state path that maximizes the conditional probability $p(y|x,\theta)$. The solution of $\max_{y\in S^n}p(y|x,\theta)$
can be found by a dynamic
programming algorithm called the {\it Viterbi algorithm}.
For a given observation sequence $x$ and an HMM with initial probabilities $(p_{0k})$,
transition matrix $\mathbb{P}=(p_{lj})$ and emission densities
$f_k(\cdot)$ ($k,l,j\in S$), the Viterbi algorithm is the following:
\begin{enumerate}
  \item[(1)] for every $k\in S$, define $\delta_1(k):=p_{0k} f_k(x_1)$;
  \item[(2)] for $t=1,\ldots,n-1$ and $k\in S$ calculate
  \[ \delta_{t+1}(k)=\max_{l \in S}\big(\delta_{t}(l)p_{lk}\big)f_k(x_{t+1}), \]
  and record
$$l_t(k):=\arg\max_{j \in S} \delta_t(j)p_{jk};$$
  \item[(3)] find the Viterbi path $v$ by backtracking:
$$v_n:=\arg\max_{k \in S} \delta_n(k),\quad v_t=l_t(v_{t+1}),\quad
t=n-1,\ldots,1.$$
\end{enumerate}
Observe that the Viterbi path is not necessarily unique, any path
$v=\arg\max_{y\in S^n} p(y|x,\theta)$ is called a Viterbi path. When
emission distributions are continuous (like Gaussian as in the
present paper), the Viterbi path is unique almost surely. Although
the Viterbi path is not optimal for minimizing the expected number
of classification errors, it is the most popular and most studied
hidden path estimate in practice (see
e.g.~\cite{rabiner,koski,bishop,CentroidEstimators2008}). The hidden
path estimate that minimizes the expected number of
 classification errors is the so-called PMAP (pointwise maximum a posteriori) path which maximizes the sum $\sum_{t=1}^n p_t(y_t|x,\theta)$. Since the sum can obviously be maximized termwise, the
 PMAP path $y$ is just a state path where $y_t$, $t=1,\ldots,n$, maximizes the marginal probability $p_t(y_t|x,\theta)$, and therefore the PMAP path can be found pointwise.
 Because of the pointwise optimization its posterior probability can be zero due to inadmissible transitions. The Viterbi and PMAP path are of different nature and for
 many models the difference between them can be rather big. For a
 discussion about the Viterbi, PMAP and related paths, see \cite{JMLR}.
%
%
\subsection{Bayesian approach} The Viterbi algorithm is applicable when the transition matrix as well
as emission parameters are known. When this is not the case, the
standard approach is to first estimate the parameters and then
perform segmentation. This approach -- parameters first, then
segmentation -- is also applicable in the Bayesian framework, where
the parameters of an HMM are considered random. Indeed, one can find
a Bayesian point estimate, typically the posterior mode
$\hat{\theta}=\arg\max_{\theta}p(\theta|x)$, and then perform
segmentation. However, if the primary goal is segmentation rather
than parameter estimation, one can consider the true underlying path
as the actual parameter of interest and the emission and transition
parameters as nuisance parameters, and perform segmentation
directly. {Let us explain that approach more formally}.
\paragraph{Bayesian Viterbi path.} Let $\pi$ be a prior density in $\Theta$  with respect to a
reference measure $d\theta$. For any $\theta=(\theta_{tr},\theta_{em})$ and for any pair $(x,y)$,
\begin{equation}\label{likef}
p(x,y|\theta)=p(x|y,\theta_{em})p(y|\theta_{tr}), \quad  p(y,x)=\int p(y,x|\theta) \pi(\theta)d \theta.
\end{equation}
It is important to note that although for any parameter set $\theta$
the measure $p(y,x|\theta)$ is a distribution of an HMM, then the
measure $p(y,x)$ obtained after mixing is a distribution of a
process that in general is not an HMM (sometimes called mixed-HMM,
see \cite{maruotti1,maruotti2}). This complication is typical in the
Bayesian setup and not specific to HMMs -- a mixture of a product
measure (the law of independent random variables) is not a product
measure anymore, a mixture of Markov chains is not a Markov chain
anymore, and so on. However, this circumstance complicates the whole
analysis. As previously, the {\it Viterbi path} $v$ is defined as
any state sequence $y\in S^n$ that maximizes the probability
$p(y|x)$ over all state sequences: $v=\arg\max_{y\in S^n} p(y|x)$.
As in the case of HMM, the Viterbi path is not necessarily unique,
although in the case of continuous emissions it typically is. It
might happen though that for many paths the probability $p(y|x)$ is
very close to maximum, and even if these paths are not formally
Viterbi paths (because they correspond to local and not to global
maximums), they might often be outputs of the iterative algorithms
considered in the article. This could seem disappointing at the
first sight that the algorithms fail to find the global maximum, but
since the conditional probabilities are very close to the maximum,
these suboptimal paths could be considered as good substitutes of
the Viterbi path.

As it is typical in Bayesian analysis, any Viterbi path $v$ is best
only on average. For a given parameter $\theta$ generated from
$\pi$, the path $v$ obviously does not need to maximize the
 probability $p(y|x,\theta)$, but it maximizes (over $y$) the
average probability:
\begin{equation}\label{possa}
p(y|x)={p(y,x)\over p(x)}={\int p(y,x|\theta) \pi(\theta)d
\theta\over p(x)}={\int p(y|x,\theta)p(x|\theta) \pi(\theta)d
\theta\over p(x)}=\int
p(y|x,\theta)p(\theta|x)d\theta,\end{equation} where $p(\theta|x)$
is the posterior probability. When the observations $x$ are
generated from a particular distribution with the true parameter
$\theta^*$ and $x$ is sufficiently long, then according to the
posterior consistency the posterior measure $p(\theta|x)$ is
concentrated around $\theta^*$, and then the Bayesian approach
should give more or less the same result as the parameter-first
approach. Therefore, the Bayesian approach is more appealing when
the sequence of observations $x$ is not very long. In particular, it
might be a very reasonable choice in HMM pattern recognition setup
when the training data consist of pairs
$(x^1,y^1),\ldots,(x^m,y^m)$. For a training pair $(x^j,y^j)$ the
observation sequence $x^j$ and the corresponding state sequence
$y^j$ are assumed to be generated from an HMM with unknown parameter
$\theta^j$. There is a target sequence $x$ whose Viterbi path needs
to be estimated. Observe that in pattern recognition the Viterbi
path ${v}$ is a common choice, because it minimizes the following
expected loss
$${v}=\arg\min_y \sum_{y'}L(y',y)p(y'|x),\quad L(y',y)=\left\{
                                                             \begin{array}{ll}
                                                               1, & \hbox{when $y=y'$;} \\
                                                               0, & \hbox{else.}
                                                             \end{array}
                                                           \right.$$
In this setup, for every $j$ the parameter estimate $\hat{\theta}^j$
can be found. When these estimates do not vary much, it is
reasonable to believe that the true parameters $\theta^j$ are the
same: $\theta^j=\theta^*$. The best one can do is to  aggregate all
estimates $\hat{\theta}^j$ into one estimate $\hat{\theta}$ which
gives a reliable estimate of $\theta^*$, and use $\hat{\theta}$ to
estimate the Viterbi path: $\hat{v}=\arg\max_y p(y|x,\hat{\theta})$.
However, when the parameter estimates $\hat{\theta}^j$ vary a lot,
it is reasonable to believe that the training data parameters
$\theta^j$ are not the same, but rather constitute a sample from a
prior distribution $\pi$. The prior $\pi$ could be chosen so that
its variance (or mean, moments, hyperparameters) matches the
variance (or mean, moments, hyperparameters) of the sample
$\hat{\theta}^1,\ldots, \hat{\theta}^m$. Assuming that the true
parameter is generated by a prior $\pi$, the best one can then do is
to find the path  that maximizes the average likelihood  as in
(\ref{possa}).
\paragraph{Prior distributions.} In this article, we assume that the number of states $K$ as well as initial
probabilities $p_{0k}$ are known {and uniform. It means that we
shall not put any prior on $K$ and initial probabilities, and
throughout the paper we take $p_{0k}=1/K$, $k=1,\ldots,K$. However,
we shall put  prior $\pi$ on the set of transition matrices and
emission parameters.} The prior $\pi$ is assumed to be such that
emission and transition parameters are independent:
$\pi(\theta)=\pi_{em}(\theta_{em})\pi_{tr}(\theta_{tr})$, where
$\pi_{em}$ and $\pi_{tr}$ are marginals. Then
\begin{equation}\label{fact}p(y,x,\theta)=p(y|\theta_{tr})\pi_{tr}(\theta_{tr})p(x|y,\theta_{em})\pi_{em}(\theta_{em}),\end{equation}
with
\begin{align*}
p(y|\theta_{tr})=
p_{0y_1} \prod_{lj}\big(p_{lj}(\theta_{tr})\big)^{n_{lj}(y)},\quad
p(x|y,\theta_{em})=
\prod_{k=1}^K\prod_{t: y_t=k} f_{k}(x_t|\theta^k_{em}),\end{align*}
where $n_{lj}(y)$ denotes the number of
transitions from state $l$ to state $j$ in the state sequence $y$.
%
In particular, (\ref{fact}) ensures that for given $y$ and $x$,
$p(y)$ depends on transition priors and $p(x|y)$ depends on emission
priors only. The independence also implies that the posterior of the
transition parameters depends only on $y$:
$p(\theta_{tr}|x,y)=p(\theta_{tr}|y)$, and that $\theta_{em}$ and
$\theta_{tr}$ are independent under posterior measure:
\begin{equation}\label{indt}
p(\theta|x,y)=p(\theta_{tr}|y)p(\theta_{em}|y,x).
\end{equation}
We consider the case where emission parameters are componentwise independent, that is
$\pi_{em}(\theta_{em})=\pi^1_{em}(\theta^1_{em})\cdots
\pi_{em}^K(\theta^K_{em})$ for $\theta_{em}=(\theta^1_{em},\ldots, \theta^K_{em})$,
which implies the independence under posterior:
\begin{equation}\label{factem}
p(\theta_{em}|x,y)=\prod_{k=1}^Kp(\theta^k_{em}|x,y).
\end{equation}
Typically transition parameters are the transition probabilities, that is
$p_{lj}(\theta_{tr})=p_{lj}$. The standard approach in this case is
to model all the rows of a transition matrix   $\mathbb{P}=(p_{lj})$
independently with the $l$-th row having a Dirichlet prior
$\text{Dir}(\alpha_{l1},\ldots,\alpha_{lK})$, see e.g.
\cite{koski,HMMbook,RobertMarin,corander,BayesNonparam, tag}. Thus,
$$\pi_{tr} \big(\mathbb{P}\big)=\pi_{tr}(p_{11},\ldots,p_{1K})\pi_{tr} (p_{21},\ldots,p_{2K})\cdots \pi_{tr} (p_{K1},\ldots,p_{KK})\propto
\prod_{lj}p_{lj}^{\alpha_{lj}-1},$$ provided
$(p_{l1},\ldots,p_{lK})\in \mathbb{S}_K$, where $\mathbb{S}_K$ is a unit simplex.
Since the rows are independent under the prior,
they are also independent under the posterior, so that for a given path $y$,
the $l$-th row has a Dirichlet distribution:
$$ p\big((p_{l1},\ldots,p_{lK})|y\big)\sim {\text Dir}(\alpha_{l1}+n_{l1}(y),\ldots,\alpha_{lK}+n_{lK}(y)).$$
Let $n_l(y)=\sum_j n_{lj}(y)$ and $\alpha_l=\sum_j\alpha_{lj}$. Under a Dirichlet prior, the marginal probability of any path $y$ can be
calculated as (see e.g. (19) in \cite{corander})
\begin{equation}\label{p(y)}
p(y)=\int p(y|\theta_{tr})\pi_{tr}(\theta_{tr})d\theta_{tr}=
p_{0y_1}\prod_l {\Gamma(\alpha_l)\over
\Gamma(\alpha_l+n_l(y))}\prod_j
{\Gamma(\alpha_{lj}+n_{lj}(y))\over
\Gamma(\alpha_{lj})}.\end{equation}
As is common in the case of Dirichlet priors (see e.g.~\cite{BayesNonparam}), we will use the factorization
$\alpha_{lj}=Mq_{lj}$, where $Q=\big(q_{lj}\big)$ is a transition
matrix and $M>0$ can be regarded as the {\it precision parameter}. Thus,
$Q$ postulates our belief about the general form of the transition matrix and $M$ shows how strongly
we believe in it: the bigger $M$, the smaller the variance of $p_{lj}$.
\subsection{Objectives of the article}
The main goals of the article are the following:
\begin{itemize}
  \item[-] To give a brief overview of the most commonly
  used Bayesian segmentation methods (segmentation performed
  with Bayesian parameter estimates, segmentation MM,
  variational Bayes, iterative conditional mode, and simulated
  annealing) and study their performance. We present the general
  ideas behind the methods and derive the formulae needed for
  applying the methods. Although in general the methods
  mentioned above are well known, we believe that in the context
  of segmentation in Bayesian HMMs the methods are not so well
  studied and understood. The contradictory results in the
  speech tagging literature mentioned in the introduction are
  evidence of that. Therefore, a comparative study of these
  segmentation methods together with the necessary formulae
  might clarify the picture and help practitioners.

  \item[-] To introduce the segmentation EM method and study its
performance in comparison to the Bayesian segmentation methods mentioned above.
Segmentation EM is a standard EM method where the path $y$ is
considered as the main parameter of interest and $\theta$ as a nuisance
parameter. Application of the segmentation EM algorithm depends very
much on the particular model studied. For example, to apply it to the so-called
triplet Markov models, see e.g.~\cite{Gorynin,Courbot,Lanchantin}, some additional assumptions are
needed. Similarly it is not clear how to apply segmentation EM
 in HMMs with infinite state spaces (hierarchical Dirichlet processes). Thus our main message
 concerning the segmentation EM algorithm is that in the case of HMMs the
segmentation EM approach is applicable (at least for the priors
considered in the paper) and works well. Among all the
non-stochastic methods we consider, segmentation
EM is the only one that iteratively maximizes $p(y|x)$, and therefore it is theoretically justified and recommended.

\item[-] To compare the performance of all the methods and to show that in Bayesian segmentation, the best non-stochastic iterative methods perform at least as well as MCMC methods such as simulated
annealing, while at the same time being computationally faster and less demanding.
\end{itemize}
%
\section{Bayesian segmentation methods}\label{methods}
\subsection{Segmentation EM}
Since our goal is to find a state sequence that maximizes $p(y|x)$,
the main parameter of interest is the hidden path rather than the
model parameters $\theta$. Therefore it is natural to change their
roles in the EM procedure in order to maximize $p(y|x)$. {Indeed --
in the traditional Bayesian EM approach the objective is to maximize
the posterior probability of parameters $p(\theta|x)$, and then $y$
is considered as the latent (or nuisance) parameter and integrated
out. In our setup, the objective is to maximize $p(y|x)$, thus
$\theta$ is considered as the latent (nuisance) parameter and
integrated out.}

We start with an initial sequence $y^{(0)}$ and then update  the
state sequences according to the following rule:
\begin{equation}\label{sQ}
y^{(i+1)}=\arg\max_y \int \ln p(y,\theta|x)p( \theta|y^{(i)},x) d
\theta=\arg\max_y \int \ln p(y,x| \theta)p( \theta|y^{(i)},x)  d
\theta.
\end{equation}
Every iteration step increases the probability $p(y|x)$ and the algorithm stops when there are no further
changes in the estimated state sequence. We call this procedure {\it segmentation EM}, the output is denoted by $\hat{v}_{\rm{sEM}}$.
\begin{lemma} Every iteration step in the segmentation EM procedure increases the posterior probability: $p(y^{(i+1)}|x)\geq p(y^{(i)}|x)$. Furthermore, the
objective function in (\ref{sQ}) can be maximized with the Viterbi
algorithm by considering the matrix $(u_{lj}^{(i)})$ and the functions $h_k^{(i)}$ as the
transition and emission parameters, where
\begin{equation}\label{seg-EM-V}
u_{lj}^{(i)}:=\exp\big[\int  \ln
p_{lj}(\theta_{tr}) p(\theta_{tr}|y^{(i)})d
\theta_{tr}\big],\quad h_k^{(i)}(x_t):=\exp\big[\int \ln
f_k(x_t|\theta_{em}^k) p(\theta^k_{em}|y^{(i)},x)d
\theta^k_{em}\big].\end{equation}
\end{lemma}
\noindent Observe that $u_{lj}^{(i)}$ and $h_k^{(i)}$ in (\ref{seg-EM-V}) depend on iteration $i$. For the sake of simplicity we will suppress $(i)$ in the notation and write $u_{lj}$ and $h_k$ in the rest of this subsection.
\begin{proof} It is well known that the standard EM algorithm
increases the likelihood at every iteration step (see, e.g.
\cite{EMbook,HMMbook}). Change the roles of $\theta$ and $y$ to
obtain $p(y^{(i+1)}|x)\geq p(y^{(i)}|x)$.
To see that $y^{(i+1)}$ can be found with the Viterbi algorithm, note that by
(\ref{likef}), (\ref{indt}) and (\ref{factem}) we have
\begin{align*}
&\int \ln p(y,x| \theta)p( \theta|y^{(i)},x)  d \theta =\int \ln
p(y|\theta_{tr})p( \theta_{tr}|y^{(i)})  d \theta_{tr}+\int \ln
p(x|y,\theta_{em}) p( \theta_{em}|y^{(i)},x) d \theta_{em}\\\notag
&=\ln p_{0y_1}+\sum_{l,j\in S}n_{lj}(y)\int \ln
p_{lj}(\theta_{tr})p( \theta_{tr}|y^{(i)}) d
\theta_{tr}+\sum_{k=1}^K \sum_{t: y_t=k}\int \ln
f_k(x_t|\theta^k_{em})p(\theta^k_{em}|y^{(i)},x)d\theta^k_{em}\\\label{u}
&=\ln p_{0y_1}+\sum_{l,j\in S}n_{lj}(y) \ln u_{lj}+\sum_{k=1}^K
\sum_{t: y_t=k} \ln h_k(x_t).\end{align*}
Thus, the objective
function is in the form of $\ln p(y,x)$ of an HMM with `transition
matrix' $(u_{lj})$ and `emission densities' $h_k$. By Jensen's
inequality we know that the rows of $(u_{lj})$
do not sum up to one, thus $(u_{lj})$ is not a
transition matrix. Similarly, the functions $h_k$ do not integrate
to one, thus the functions $h_k$ are not probability densities.
However, the Viterbi algorithm can still be applied to find the path
with maximum probability. To see that, note first that the functions
$h_k$ enter into the Viterbi algorithm only via values $h_k(x_t)$, so
it really does not matter whether they integrate to one or not.
Similarly, the optimality principle -- if a maximum probability path
passes state $k$ at time $m$, the first $m$ elements of that
path  must form a maximum likelihood path  amongst those paths  that
end in state $k$ at time $m$ -- does not depend on whether the
probabilities sum up to one or not. If the optimality principle holds,
then the Viterbi algorithm as a dynamic programming algorithm finds the
maximum probability path.
\end{proof}
The fact that the Viterbi algorithm can be applied for
maximizing $\ln p(y,x)$ makes the segmentation EM possible as soon
as $(u_{lj})$ and $h_k(x_t)$ can be calculated. In our case with Dirichlet priors for the transition parameters the
posterior measure $p(\theta_{tr}|y)$ is the product of the row
posteriors, and the posterior of the $l$-th row is ${\text
Dir}(\alpha_{l1}+n_{l1}(y),\ldots,\alpha_{lK}+n_{lK}(y)).$ Then
$$p_{lj}\sim {\rm Be}\big(\alpha_{lj}+n_{lj}(y),
\alpha_l+n_l(y)-\alpha_{lj}-n_{lj}(y)\big).$$ It is known that
when $X\sim {\rm Be}(\alpha,\beta)$, then
$E(\ln X)=\psi(\alpha)-\psi(\alpha+\beta),$ where $\psi$ is the digamma
function. Thus, for any sequence $y$, the quantities $u_{lj}$ can be
calculated with the following formula:
\begin{equation}\label{uij}
\ln u_{lj}(y)=\int \ln
p_{lj}(\theta_{tr})p(\theta_{tr}|y)d\theta_{tr}=\psi(\alpha_{lj}+n_{lj}(y))-\psi(\alpha_l+n_l(y)).\end{equation}
Computing $h_k$ depends on the family of emission densities.
If emission distributions belong to an exponential family, that is
$$f(x|\theta_{em})=\exp[\theta_{em}^t  T(x)+A(\theta_{em})+B(x)],$$
then calculation of $h_k(x)$ reduces to evaluating the moments
$$\int \theta_{em}  p(\theta_{em}|y,x) d\theta_{em},\quad \int A(\theta_{em})
p(\theta_{em}|y,x)d\theta_{em}.$$ For conjugate priors, this kind of
integration is often feasible.
%
%
\subsection{Other segmentation methods}\label{sec:MM}
\paragraph{Segmentation MM.}
\noindent The segmentation MM algorithm is just like the
segmentation EM algorithm, except that the expectation step is
replaced by the maximization step. We start with an initial
path $y^{(0)}$.  Then, given $y^{(i)}$, find
 $$\theta^{(i+1)}=\arg\max_{\theta} p(\theta|y^{(i)},x),\quad y^{(i+1)}= \arg\max_y
p(y|\theta^{(i+1)}, x).$$  The algorithm converges when there are no changes in the two consecutive path estimates.
Every iteration step increases the joint likelihood, that is
$$p(y^{(i+1)},\theta^{(i+1)}|x)\geq p(y^{(i)},\theta^{(i+1)}|x)\geq
p(y^{(i)},\theta^{(i)}|x),$$ but the objective function $p(y|x)$ is
not guaranteed to increase. In the context of parameter estimation
in the non-Bayesian setting (that is when the prior is non-informative) this
algorithm is sometimes called the {\it Viterbi training}
\cite{intech,AVT1,AVT3,AVT4} or {\it classification EM}
\cite{clEM1,clEM2}. It should move on faster than segmentation EM.
The advantage of the segmentation MM procedure over the segmentation EM
procedure is that it does not require calculation of $u_{lj}$
and $h_k$. For given $\theta^{(i)}$ the path $y^{(i)}$ can be
found by the standard Viterbi algorithm, and $\theta^{(i+1)}$ is just the
posterior mode;  in our case the mode
for the emission and transition parameters can be calculated separately due to independence:
$\theta^{(i+1)}=(\theta_{tr}^{(i+1)},\theta_{em}^{(i+1)})$, where
$$\theta_{tr}^{(i+1)}=\arg\max_{\theta_{tr}}p(\theta_{tr}|y^{(i)}),\quad
\theta_{em}^{(i+1)}=\arg\max_{\theta_{em}}p(\theta_{em}|y^{(i)},x).$$
\paragraph{Bayesian EM.}
\noindent The {\it parameters-first} approach in segmentation consists
of estimating the unknown model parameters first and then performing segmentation.
The most common parameter estimate in the Bayesian setup is the {\it MAP estimate} defined as
$$\hat{\theta}=\arg\max_{\theta}p(\theta|x)=\arg\max_{\theta}p(x|\theta)\pi(\theta).$$
The standard method for calculating
$\hat{\theta}$  is the EM algorithm  \cite{EMbook,koski}. The EM procedure in the Bayesian setup starts with an initial
parameter $\theta^{(0)}$ and updates the parameters iteratively as follows:
\begin{equation}\label{parEM}
\theta^{(i+1)}=\arg\max_{\theta} \sum_y \ln p(y,\theta|x)
p(y|\theta^{(i)},x)=\arg\max_{\theta}\Big[ \sum_y \ln
p(y,x|\theta)p(y|\theta^{(i)},x)+\ln \pi(\theta)\Big].\end{equation}
Every iteration increases the posterior probability, that is
$p(\theta^{(i+1)}|x)\geq p(\theta^{(i)}|x).$ We call this estimation
procedure {\it Bayesian EM} and denote the resulting
parameter estimate by $\hat{\theta}_{\rm{B(EM)}}$. The EM procedure
in the non-Bayesian setup is the same, except that $\ln \pi(\theta)$ is missing on the right hand
side of (\ref{parEM}). This procedure
will be called {\it standard EM} and the output of the procedure
will be denoted by $\hat{\theta}_{\rm{EM}}$. Thus, the
standard EM algorithm can be considered as a special case of the
Bayesian EM algorithm with a non-informative prior ($\ln
\pi(\theta)=const$). In the case of Dirichlet transition
priors, noninformative priors correspond to the case
$\alpha_{lj}=1$. The Viterbi path estimates $\hat{v}_{\rm{B(EM)}}$ and $\hat{v}_{\rm{EM}}$ are obtained by
applying the Viterbi algorithm  with the respective parameter estimates:
$\hat{v}_{\rm{B(EM)}}:=\arg\max_y p(y|x,\hat{\theta}_{\rm{B(EM)}})$, $\hat{v}_{\rm{EM}}:=\arg\max_y p(y|x,\hat{\theta}_{\rm{EM}})$.
%
\paragraph{Variational Bayes approach.}
\noindent The idea behind the variational Bayes (VB)
 approach (see, e.g. \cite{watanabe,quinn,VBtutorial,VBtutorial2,beal, beal2, bishop})  is to approximate the posterior $p(\theta,y|x)$
with a product $\hQt(\theta)\hQy(y)$, where $\hQt$ and $\hQy$ are
probability measures on the parameter space and $S^n$ that
minimize the Kullback-Leibler divergence $D\big(\Qt\times
\Qy||p(\theta,y|x)\big)$ over all product measures $\Qt\times \Qy$, that is
$$\hQt\times \hQy =\arg\inf_{\Qt\times \Qy} D\big(\Qt\times
\Qy||p(\theta,y|x)\big).$$ It is known that the measures $\hQt$ and $\hQy$ satisfy
the equations
\begin{align*}
\ln \hQt(\theta)&=c_1+\int \ln  p(\theta,y|x) \, \hQy(dy)=c_1+\sum_y  \ln  p(\theta,y|x) \hQy(y),\\
\ln \hQy(y)&=c_2+\int  \ln p(\theta,y|x) \, \hQt(d\theta),
\end{align*}
where $c_1$ and $c_2$ are constants. This suggests the following
iterative algorithm for calculating $\hQt(\theta)$ and $\hQy(y)$.
Start with an initial sequence $y^{(0)}$ and take
$\Qy^{(0)}=\delta_{y^{(0)}}$. Given $\Qy^{(i)}$, update the measures as
\begin{align*}
\ln \Qt^{(i+1)}(\theta)&=c_1^{(i+1)}+\sum_y  \ln p(\theta,y|x)
 \Qy^{(i)}(y),\\
\ln \Qy^{(i+1)}(y)&=c_2^{(i+1)}+\int   \ln p(\theta,y|x)
\Qt^{(i+1)}(d\theta).
\end{align*}
In \cite{beal}, the algorithm is called {\it variational Bayes EM}
and it is argued (Theorem 2.1) that it decreases the Kullback-Leibler divergence in the following sense:
\[D\big(\Qt^{(i)}\times \Qy^{(i)}||p(\theta,y|x) \big)\geq D\big(\Qt^{(i+1)}\times \Qy^{(i)}||p(\theta,y|x)\big)\geq D\big(\Qt^{(i+1)}\times
\Qy^{(i+1)}||p(\theta,y|x) \big).\]
Suppose the VB algorithm described above has converged and its final output is  $\hQt\times \hQy$.
Then $\hQy$ is taken as the approximation of $p(y|x)$ and the Viterbi path estimate $\hat{v}_{\rm{VB}}$ is obtained as
$\hat{v}_{\rm{VB}}:=\arg\max_y \hQy(y)$.

Applying the variational Bayes method for estimating the Viterbi path is certainly not a trivial task. All the formulae needed for
updating $\Qy^{(i+1)}$ and $\Qt^{(i+1)}$ with explanations about technical details are presented in the Appendix.
%
%
%
\paragraph{Simulated annealing.}
\noindent Let $1\leq \beta_1<\beta_2\ldots <\beta_r$ be a {cooling
schedule}. Since direct sampling from distribution
$p_{\beta}(y|x)\propto p^{\beta}(y|x)$ is not possible, for every
$\beta$ we sample
$y_{\beta}^{(1)},\theta_{\beta}^{(1)},y_{\beta}^{(2)},\theta_{\beta}^{(2)},\ldots,y_{\beta}^{(n_{\beta})}$
alternately from a probability measure $p_{\beta}(\theta,y|x)\propto
p(\theta,y |x)^{\beta}$ in the acceptance-rejection sense as
follows. For given $\beta$ and path $y^{(i)}$, generate the
parameter $\theta^{(i)}$ from the distribution
$p_{\beta}(\theta|y^{(i)},x)\propto p(\theta|y^{(i)},x)^{\beta}$.
Then, given $\theta^{(i)}$, generate a path $y$ from
$p_{\beta}(y|\theta^{(i)},x)\propto p(y|\theta^{(i)},x)^{\beta}$.
The generated path $y$ will be accepted as $y^{(i+1)}$ with the
probability
\begin{align*}
 \frac{  p( y|x)^{\beta} /[ p_{\beta}( y | \theta^{(i)},x ) p_{\beta}( \theta^{(i)}  | y^{(i)},x)]}
  {  p( y^{(i)}|x)^{\beta}/[p_{\beta}( y^{(i)} | \theta^{(i)},x ) p_{\beta}( \theta^{(i)}  | y,x)]} \wedge 1
 = \frac{p(y|x)^{\beta}/p_{\beta}(y,x)}{p(y^{(i)}|x)^{\beta}/p_{\beta}(y^{(i)},x)}\wedge 1,
\end{align*}
where $J_{\beta}(y|y^{(i)})=p_{\beta}( y | \theta^{(i)},x ) p_{\beta}( \theta^{(i)}  | y^{(i)},x)$ is the proposal distribution and
$p_{\beta}(y,x)\propto \int p(y,x|\theta)^{\beta}\pi(\theta)^{\beta}d\theta$. Note that the
ratio actually does not depend on $\theta^{(i)}$. If the candidate path $y$ is not accepted, then a new
 parameter $\theta^{(i)}$ from $p_{\beta}(\theta|y^{(i)},x)$ and a new path $y$ from the distribution
 $p_{\beta}(y|\theta^{(i)},x)$ will be generated.
At the end of the sampling, the path with highest probability is
found:
$$\hat{v}_{\rm SA}:=\arg\max_{k=1,\ldots,r;\, i=1,\ldots,n_{\beta_k}} p(y_{\beta_k}^{(i)}|x).$$
\paragraph{Iterative conditional mode algorithm.}
\noindent As already mentioned, sampling from $p(y|x)$ is in
general not possible even if the model is simple. Since
for any path $y$ the probability $p(y|x)$ can be found, then also for any site $t$ the probability $p_t(y_t|y_{-t},x)$ can
be calculated, where  $p_t(y_t|y_{-t},x)$ stands for the probability of
observing $y_t$ at site $t$ given the rest of the sequence and $x$.
Note that because $p(y|x)$ is not a Markov measure, $p_t(y_t|y_{-t},x)$
is not necessarily the same as $p_t(y_t|y_{t-1},y_{t+1},x)$. The
{\it iterative conditional mode (ICM)} updates paths iteratively as follows. It starts from a sequence $y^{(0)}$. To obtain $y^{(i+1)}$, the sequence $y^{(i)}$ is updated site-by-site
by the following rule:
$$ y_t^{(i+1)}=\arg\max_{k\in S} p_t(k |y^{(i+1)}_1,\ldots,
y^{(i+1)}_{t-1},y^{(i)}_{t+1},\ldots,y^{(i)}_n,x).$$ Thus, the ICM algorithm acts similarly to single site sampling (\cite{HMMbook}, \cite{johnson2}), but instead of generating a random state,
at every step it picks
 a state with maximum probability. In
\cite{corander}, the ICM algorithm is used under the name `greedy
algorithm'. It is indeed greedy in the sense that the  update of
every site increases the probability $p(y|x)$. The ICM algorithm
converges when no further changes occur in the estimated sequence;
the output will be denoted by $\hat{v}_{\text{ICM}}$.

It is well known from the theory of simulated annealing that such a greedy update can cause the output to be trapped in a local
maximum (see, e.g. \cite{winkler}), and our numerical examples confirm that. However, since \cite{corander} is one of the
few papers that considers segmentation in the Bayesian framework by non-stochastic methods, we include this method in our study.
%
%
\section{Numerical examples} \label{simulations}
To illustrate the behaviour of the segmentation methods described in Section 2, we will present the results of two examples. In the first example we study the case with known emission distributions and transition probabilities following Dirichlet priors. Thus, $\theta=\theta_{tr}$, $\pi=\pi_{tr}$ and under $\pi$, the rows of
the transition matrix are independent with the $l$-th row having a Dirichlet distribution ${\rm Dir}(\alpha_{l1},\cdots,\alpha_{l4})$. In the second example emission parameters also are assumed to be
unknown and normal emissions with conjugate priors are studied. Since the estimation criterion is $\arg\max_y p(y|x)$, the main measure of goodness is $p(y|x)$ or equivalently, $\ln p(y,x)$. We will also study how the methods perform in regard to initial state sequences and how the estimated state paths depend on different sets of prior parameters.

\subsection{General framework}
The data is generated from an HMM with four underlying states, thus
$S=\{1,2,3,4\}$. The emission distributions are normal with common
variance $\sigma^2=0.25$, the emission distribution corresponding to
state $k$ is ${\cal N}(\mu_k,\sigma^2)$ with $\mu_1=-0.7$,
$\mu_2=0$, $\mu_3=0.7$ and $\mu_4=1.4$, respectively. The transition
matrix is given by $\mathbb{P}=(p_{lj})$ with $p_{ll}=0.6$,
$l=1,\ldots,4$, and $p_{lj}=0.4/3$, otherwise. The initial
distribution $(p_{0k})$ is given by $p_{0k}={0.25}$, $k=1,2,3,4$. The length
of the generated data sequence $x$ is $n=600$.
%
\paragraph{Hyperparameters.} {Recall that we use the parametrization $\alpha_{lj}=Mq_{lj}$, $l,j=1,\ldots,4$, where
$Q=\big(q_{lj}\big)$ is a transition matrix and $M>0$ the precision
parameter.}  We will consider three $Q$-matrices:
\[ Q_1=(q_{lj}) \,\, \mbox{with} \,\, q_{lj}=0.25 \,\, \forall l,j; \quad Q_2=\mathbb{P}; \quad Q_3=(q_{lj}) \,\, \mbox{with} \,\, q_{ll}=0.4, \,\, q_{lj}=0.2 \,\, \mbox{for} \,\, l\neq j.  \]
Thus, the combination $Q_1$ and $M=4$  corresponds to uniform priors
on transition parameters, and $Q_1$ together with very large $M$ puts
a uniform prior $p(y)$ on sequences (see also Section 4.1). The matrices $Q_2$ and $Q_3$
favour sequences with long blocks; the smaller $M$ is, the more such
behavior is pronounced. To explain our choices of $M$ in
simulations, let us give some intuition about the role of $M$ in
some procedures. First, the Bayesian EM updates (\ref{BEMsolution})
for this parametrization are given by
\begin{align}\label{solution3EM}
{p}^{(i+1)}_{lj}&={\xi^{(i)}(l,j)+(Mq_{lj}-1)\over \sum_j
\xi^{(i)}(l,j)+(M-K)},\end{align}
where $\xi^{(i)}(l,j)$ is the expected number
of transitions from state $l$ to $j$ at iteration $i$, which varies between $0$ and $n-1$. If all
transitions are equally likely, with our $n=600$ it is
approximatively of order 37. If $M$ is much larger than $n$, then the influence of data in (\ref{solution3EM}) is
negligible and the output of the procedure is very close to $Q$. On
the other hand, a necessary condition in (\ref{solution3EM}) is that
$Mq_{lj}> 1$, which gives a lower bound to $M$.
A similar argument holds for segmentation EM. Since for any integer $n$ large enough (see e.g.~\cite{johnson}),
$\psi(n) \approx \ln(n-0.5)$, where $\gamma\approx 0.577$,
we can for large $n\ll m$ use the approximation
$\psi(m)-\psi(n) \approx \ln(m-0.5)-\ln(n-0.5)$.
Disregarding the fact that $Mq_{lj}$ might
not be an integer, (\ref{uij}) gives that for a given state sequence $y$,
\[ u_{lj}(y)\approx {Mq_{lj}+n_{lj}(y)-0.5\over M+n_l(y)-0.5}.\]
If $M$ is very small in comparison to $n_{lj}$, then $u_{lj}\approx
{n_{lj}(y)\over n_{l}(y)}$ and the segmentation EM algorithm
is practically the same as the segmentation MM algorithm. If on the
other hand $M$ is too big, then the data are negligible and the
output is close to the Viterbi path with $Q$.
Based on these arguments, we consider the following constants $M$: 600, 150, 50, 10, 5.
Observe that segmentation MM and Bayesian EM are applicable when $Mq_{lj}> 1$, which is restrictive when hyperparameters
$\alpha_{lj}\leq 1$ are of interest.

\paragraph{Initial sequences.} Since the  non-stochastic
methods studied here depend on initial path values, the choice of
initial paths has an important role in our numerical examples. All
our procedures are designed to start with initial sequence, but a
closer inspection of formulae $(\ref{uij})$, (\ref{MM})  and
(\ref{gamma0}) reveals that when emission parameters are known and
transition probabilities have Dirichlet priors, then the
segmentation EM, segmentation MM, Bayesian EM and variational Bayes
algorithms actually depend on $y^{(0)}$ only through the frequency
matrix or empirical transition matrix $(n_{lj}(y^{(0)}))$. The only
deterministic algorithm that uses the full initial sequence as
information and not only its summary measure through the number of
empirical transitions is ICM. Therefore, it is expected that ICM is
more sensitive with respect to initial sequences, because there are
many more actual sequences than frequency matrices. For MCMC methods
such as simulated annealing the initial value does not matter,
because the number of sweeps is typically large.

Since our goal is to find the global maximum of $p(y|x)$ and the
output of a method depends typically on the frequency matrix of the
initial sequence, we try to choose initial sequences so that the
corresponding frequency matrices will be different. Theoretically we
would somehow like to cover the whole space of transition matrices.
In the simplest case -- that is,  for a two-state model -- we could for
example choose transition matrices as {follows:
\begin{align*}
 \left(
              \begin{array}{cc}
                p & 1-p \\
                1-q & q \\
              \end{array}
            \right), \quad \mbox{where} \quad p,q\in \{ 0.25,0.5,0.75 \}. \end{align*}
This would provide us with nine different transition matrices which could then be used to generate random sequences as
realizations of a Markov chain with initial distribution being the
stationary one. In the case of four states applying the described
approach becomes more complicated. Therefore, in our examples we
have considered 15 transition matrices $B_1,\ldots,B_{15}$ for
generating initial sequences, which are obtained as follows. The
first three matrices are just our $Q_1$, $Q_2$ and $Q_3$. The rest,
$B_{4},\ldots,B_{15}$, have been randomly generated: each
row of $B_l$ has been independently generated from ${\rm Dir}(\alpha,\alpha,\alpha,\alpha)$, where the following 12
constants $\alpha=(0.3,0.5,0.7,0.8,0.9,1,1.1,1.2,1.3,1.5,1.7,1.9)$
have been used. From each matrix we have generated three random sequences as realizations of a Markov chain with initial
distribution being the stationary one. We study also
the initial sequence $y^{(0)}$ that corresponds to maximizing emissions pointwise, that is
\[ y^{(0)}_t=\arg\max_{k=1,2,3,4}f_k(x_t), \quad t=1,\ldots,n,\]
which was suggested in \cite{corander}. For given $Q$, a good
candidate for initial path is always the Viterbi path obtained using
$Q$, therefore the last initial sequence considered is the Viterbi
path obtained with the transition matrix $Q$. Thus all together we
have studied 47 initial sequences. Given a set of hyperparameters
and a non-stochastic iterative method, every initial sequence
produces an output sequence. The maximum number of different output
sequences is 47. The smaller that number, the more {\it robust} or
{\it less sensitive} with respect to initial sequences the method
is. The end result of the method is given by the best output
sequence, i.e.~the one that has the largest log-likelihood $\ln
p(y,x)$. {Since the emission densities $f_k$, $k=1,\ldots,K$, are
fixed, the optimality criterion  $\ln p(y,x)$ can for a given state
path $y$ be calculated as
\begin{equation}\label{optcr}
\ln p(x,y)=\ln p(y)+\ln p(x|y)=\ln p(y)+\sum_{t=1}^n \ln f_{y_t}(x_t),
\end{equation}
where $\ln p(y)$ is as in (\ref{p(y)}).}
\subsection{Example 1: fixed emission distributions} \label{Imodel}
The main purpose of the first example is to compare the general
performance of the  algorithms. All non-stochastic methods
(segmentation EM, segmentation MM, VB, ICM, Bayesian EM, standard
EM) were run with all 47 initial sequences, whereas for simulated
annealing one initial sequence was used. In the case of segmentation
EM, segmentation MM and ICM the algorithm stopped when there were no
further changes in the estimated state sequence. In simulated
annealing a cooling schedule with inverse temperatures equally
spaced in the range [1,10.2] was used, where for every inverse
temperature 15 paths were generated. Observe that for $Q_2$ and
$Q_3$ segmentation MM and Bayesian EM were not applicable for $M=5$
($Mq_{lj}>1$ is not fulfilled for all $q_{lj}$), therefore the
respective cell values of the tables summarizing the results for
different methods are `na'.

In Table \ref{ex1logpyx}, for every method the best log-likelihood
value $\ln p(\hat{v}, x)$ over the outcomes corresponding to 47
initial sequences is presented, where $\hat{v}$ denotes the
best output sequence for the corresponding method. The number in the brackets gives the number of different
outputs out of 47 possible. The best
log-likelihood value over all the methods for each set of
hyperparameters is given in bold.
As the table shows, the best results are generally
obtained by segmentation EM and segmentation MM. The
results for those methods differ for five sets of hyperparameters,
and then sometimes the segmentation EM performs slightly better and
sometimes the other way around. The similarity of the segmentation
EM and MM methods is explained in Section \ref{disc}.
In Table \ref{ex1logpyx}, we can also see that VB and Bayesian EM
behave quite similarly; this will also be clarified in Section \ref{disc}. Observe that
Bayesian EM is independent of initial sequence, while VB can result
in different path outcomes.  We can see that for EM-type methods the
number of different outputs (sensitivity) increases
when $M$ decreases and this makes sense, because a smaller $M$ means
that data has more influence. Notice that ICM is the most sensitive
among the studied methods, resulting in a different outcome for
basically every initial sequence. It can also be remarked that the
number of different outputs in the table for segmentation EM and
segmentation MM shows that the initial state sequences generated
from the same transition matrix result often in different output
sequences.

\begin{table}[h!]
\begin{center}
{\small
\begin{tabular}{| c | c || c| c | c | c | c | c | c |}
  \hline
   $Q$ & $M$ & sEM  & sMM & ICM & VB  & B(EM)   & EM &  SA  \\
  \hline\hline
 $Q_1$ & 600 & \textbf{-1071.76} (6)   & \textbf{-1071.76}  (6) & \textbf{-1071.76}   (11) & -1072.93 (2) & -1072.98 & -1127.27    (4)& -1072.02 \\
       & 150 &\textbf{-1017.57}   (30) & -1017.59   (27) & -1030.48   (46) & -1051.82 (4) &-1051.92 & -1038.00    (4) & -1031.63 \\
        & 50  & -940.46   (31) &  \textbf{-940.42}   (31) &  -971.08   (46) & -1019.84 (8) &-1019.07 &  -955.78    (4)&  -945.68 \\
        & 10  & \textbf{-861.81}   (38) &  -861.84   (37) &  -909.67   (46) & -932.17 (4)  & -924.27 & -878.13    (4) &  -865.10 \\
      & 5   & \textbf{-842.27}   (33) &  -842.30   (35) &  -901.81   (46) & -909.12 (1) & -899.53 & -863.33    (4) &  -860.88 \\
\hline
$Q_2$ & 600 & \textbf{-898.31}   (1) &  \textbf{-898.31}   (1)   & -898.78   (47) & -899.28 (1) & -899.28 & -927.24    (4) &  \textbf{-898.31} \\
       & 150 & \textbf{-882.51}    (8) &  \textbf{-882.51}    (7) &  -887.10   (47) & -888.32 (2) & -887.35 & -901.07    (4) &  -882.68 \\
       & 50  & \textbf{-862.31}   (18) &  \textbf{-862.31}   (19) &  -877.19   (47) & -878.32 (4) & -877.91 & -876.46    (4) &  -865.63 \\
       & 10  & \textbf{-831.71}   (32) &  \textbf{-831.71}   (36) &  -873.11   (47) & -871.52 (3) & -869.18 & -853.36    (4) &  -851.10 \\
      & 5   & \textbf{-825.07}   (36) &  na             &  -875.43   (47) & -870.10 (2) & na & -849.63    (4) &  -834.38 \\
\hline
$Q_3$ & 600 & \textbf{-985.91}    (6) &  \textbf{-985.91}    (6) &  -988.46   (47) & -989.51 (1) & -989.22 &-1010.59    (4) &  -985.91 \\
       & 150 & \textbf{-945.93}   (13) &  \textbf{-945.93}   (14) &  -961.60   (47) & -966.80 (4) & -966.80 & -964.31    (4) &  -946.58 \\
       & 50  & -901.42   (24) &  \textbf{-901.38}   (27) &  -936.78   (47) & -940.26 (2) & -938.31 & -913.28    (4) &  -905.46 \\
       & 10  & \textbf{-846.64}   (34) &  \textbf{-846.64}   (32) &  -901.30   (47) & -895.44 (1) & -892.72 & -865.09    (4) &  -865.05 \\
      & 5   & \textbf{-833.49}   (34) & na             & -895.74   (47) & -888.02 (4) & na & -855.95    (4) &  -843.05 \\
 \hline
\end{tabular} }
\end{center}
\caption{\label{ex1logpyx} The best log-likelihood value {$\ln
p(x,\hat{v})$ (calculated as in (\ref{optcr}))} obtained for every
method in Example 1. The best result(s) for every set of
hyperparameters is presented in bold. In the brackets, the number of
different output sequences out of 47 possible is given.}
\end{table}

The log-likelihood values in Table \ref{ex1logpyx} give a
general summary measure for comparing the best paths over the
 methods. To understand better how different these best paths
really are, we have counted the pointwise differences in comparison
to the best path and summarized these in Table \ref{ex1pointDiff}.
If the best state path over all the methods is $\hat{v}$ and the
best path for a method we want to compare it with is $y$, then the
sum of pointwise differences is given by $\sum_{t=1}^n I\{ \hat{v}_t
\ne y_t\}$. We can see that in the worst case, the path estimates
can differ from the best path in up to 1/3 of the path points, see
VB and Bayesian EM for $Q_1$ and $M=50$.
\begin{table}[h!]
\begin{center}
\begin{tabular}{| c | r || r| r | r | r | r |r|r |}
  \hline
   $Q$ & $M$ & sEM & sMM &  ICM & VB  & B(EM)  & EM & SA  \\
  \hline\hline
$Q_1$ & 600 &  0 &  0 &  0 & 19 &   20 & 167 & 10 \\
        & 150 &  0 &  5 & 75 & 178 &  177 & 133 & 51 \\
        & 50  &  6 &  0 &119 & 198 &  198 & 108 & 24 \\
        & 10  &  0 &  1 &135 & 158 &  150 & 79 & 114 \\
       & 5   &  0 &  2 &151 & 141 &  135 & 81 &  87 \\
\hline
$Q_2$ & 600 &  0 &  0 & 12 & 19 &   19 & 117 &  0 \\
        & 150 &  0 &  0 & 60 & 53 &   47 & 117 &  6 \\
        &  50 &  0 &  0 & 92 & 66 &   67 & 101 & 42 \\
        &  10 &  0 &  0 &170 & 113 &  110 & 81 & 86 \\
      & 5   &  0 & na &170 & 118 &  na & 81 & 28 \\
\hline
$Q_3$ & 600 &  0 & 0 & 38 & 41 &   37 & 90 &  9 \\
        & 150 &  0 & 0 & 85 & 92 &   92 & 82 & 15 \\
        & 50  &  7 & 0 &130 & 125 &  122 & 83 & 163 \\
        & 10  &  0 & 0 &160 & 138 &  132 & 80 & 44 \\
        & 5   &  0 & na & 151 & 129 & na &  81 & 54 \\
\hline
\end{tabular}
\end{center}
\caption{\label{ex1pointDiff} Comparison of the estimated state
sequences with the best Viterbi path estimate for each set of
hyperparameters in Example 1. The number of pointwise differences
compared to the best path estimate is presented.}
\end{table}

{Tables \ref{ex1logpyx} and \ref{ex1pointDiff} summarize the results
of different methods for a fixed observation sequence $x$.} To
observe the general behavior of the  algorithms, we have rerun these
simulations  for 20 different observation sequences. Different
observation sequences show a similar pattern to that in Table
\ref{ex1logpyx}. We now study how segmentation EM and segmentation
MM perform in comparison to simulated annealing. Table
\ref{n600sum20seq} summarizes the results for our 15 sets of
hyperparameters and 20 observation sequences. The counts in the
first half of the table (columns SA$_{max}$, sEM$_{max}$,
sMM$_{max}$) present for each set of hyperparameters the number of
best scores over 20 observation sequences. Each time a  method is
counted as best when it reaches the maximum log-likelihood for a
given set of hyperparameters. Thus, for example, if all the three
methods resulted in the same state path estimate, every method is
counted as the best or `winner'. In our example we can conclude that
EM-type methods perform better than simulated annealing, since
simulated annealing does not give the maximum log-likelihood value
as often as EM-type methods. The second part of the table (columns
SA$_{min}$, sEM$_{min}$, sMM$_{min}$) presents for each method and
for a given set of hyperparameters the count over 20 observation
sequences of when this method was strictly worse than the other two.
Now we want to identify a `loser', therefore we have counted how
many times the respective method performs worst of the three methods
according to the log-likelihood value. Here we can see that
simulated annealing gives the lowest log-likelihood value most
often. Thus, our example demonstrates that for a given cooling
schedule and given set of initial state paths, EM-type algorithms
perform actually better than simulated annealing.
%
%
%
\begin{table}[h!]
\begin{center}
\begin{tabular}{| c | c || c| c | c || c | c | c |}
  \hline
   $Q$ & $M$ & SA$_{max}$  & sEM$_{max}$ & sMM$_{max}$ & SA$_{min}$  & sEM$_{min}$   & sMM$_{min}$  \\
  \hline\hline
$Q_1$ & 600  &   0 &    19 &    20 &   20  &    0  &    0 \\
      & 150  &   0 &    18 &    15 &   20  &    0  &    0 \\
      &  50  &   0 &    16 &    16 &   20  &    0  &    0 \\
      &  10  &   1 &    14 &    14 &   18  &    1  &    0 \\
      &   5  &   5 &    11 &    12 &   14  &    2  &    2 \\
\hline
$Q_2$ & 600  &  16 &    20 &    20 &    4  &    0  &    0 \\
       &  150  &   1 &    18 &    20 &   19  &    0  &    0 \\
       &   50  &   2 &    16 &    15 &   18 &     1  &    1 \\
       &  10  &    4 &    11 &   17 &    15 &    2  &    0 \\
       &   5  &    4 &    16 &    na &   16 &     4  &   na \\
\hline
$Q_3$ & 600  &   6 &    17 &    18 &   13 &     1  &    0 \\
       &  150  &   0 &    19 &    20 &   20 &     0  &    0 \\
       &   50  &   1 &    17 &    16 &   19 &     0  &    0 \\
       &   10  &   5 &    13 &    11 &   14 &     1  &    4 \\
      &    5  &   3 &    17 &     na &   17 &     3  &   na \\
\hline
\end{tabular}
\end{center}
\caption{\label{n600sum20seq} The counts over the Viterbi path estimates corresponding to 20 different observation sequences showing when the simulated annealing, segmentation EM and
segmentation MM methods reached the maximum and minimum values of log-likelihood in Example 1.
The minimum count shows how many times the respective method performs worst of the three methods according to the log-likelihood value.}
\end{table}
\subsection{Example 2: priors on transition probabilities and emission parameters} \label{caseII}
In the second example we assume that  the parameters of emission
densities are also unknown. The transition probabilities are modeled with Dirichlet priors as before. For emissions we consider
normal distributions with conjugate prior distributions. The emission distribution corresponding to state $k$ is
$\mathcal N(\mu_k,\sigma_k^2)$, where prior distributions for $\mu_k$ and $\sigma_k^2$ are given by a normal and
inverse chi-square distribution respectively (also known as NIX priors):
$$\pi_{em}(\theta_{\rm em})=\prod_{k=1}^K\pi_{em}(\theta^k_{\rm em}),\quad
\pi_{em}(\theta^k_{\rm em})=\pi(\mu_k,\sigma_k)=\pi(\sigma^2_k)\pi(\mu_k|\sigma_k^2), $$
where
$$\mu_k |\sigma_k^2\sim \mathcal N \big(\xi_k,{\sigma_k^2\over \kappa_0}\big),\quad
\sigma^2_k\sim {\rm Inv-}\chi^2(\nu_{0},\tau^2_{0}).$$ Here
$\kappa_0$, $\nu_{0}$ and $\tau^2_{0}$ are hyperparameters that
might depend on $k$, but  in our example we assume they are  equal.
The calculations have been performed using the same 20 observation
sequences $x$ and the same 47 initial sequences $y^{(0)}$ as in
Example 1. We will also refer to this example as the Dirichlet-NIX
case. The necessary formulae and computational details about the
algorithms needed for the Dirichlet-NIX example can be found in the
Appendix. {As previously, for any path $y$ the universal optimality
criterion is  $\ln p(x,y)=\ln p(y)+\ln p(x|y)$, where $p(y)$ is
calculated by (\ref{p(y)}) and under a NIX-prior $\ln p(x|y)$ is
calculated as in (\ref{jaana}).}

In the Dirichlet-NIX example, the choice of emission hyperparameters
affects segmentation results strongly, {see also Subsection
\ref{clusters}}. The hyperparameters we consider are as follows:
$\xi=(-0.7,0,0.7,1.4)$, $\tau_0^2=0.25$, $\kappa_0=10$, $\nu_0=50$.
In simulated annealing a cooling schedule with inverse temperatures
equally spaced in the range $[1,21]$ was used, for every inverse
temperature 15 paths were generated. Again, for $Q_2$ and $Q_3$
segmentation MM and Bayesian EM were not applicable for $M=5$
($Mq_{lj}>1$ is not fulfilled for all $q_{lj}$), therefore the
respective cell values of the tables summarizing the results for
different methods are `na'.

Since we are very much interested in how much faster non-stochastic methods perform computationally in comparison to MCMC methods,
 we start with presenting in Table \ref{n600sum20seqEx2} a summary of the behaviour of log-likelihood values over the 20 sequences and our 15 sets of transition hyperparameters
{just as in Table \ref{n600sum20seq}}. The counts in Table \ref{n600sum20seqEx2} show that simulated annealing often gives
 the maximum log-likelihood value for $M=600$, otherwise
EM-type algorithms perform generally better.

In Table \ref{ex2logpyx}, the log-likelihood values $\ln p(\hat{v},x)$ of the best path estimates for each  method
are presented for the same observation sequence as in Table \ref{ex1logpyx}.
Again, the number of different outcome sequences out of 47 possible
can be seen in the brackets. In general Table \ref{ex2logpyx} shows
the same pattern as Table \ref{ex1logpyx}: the best methods are
segmentation MM and segmentation EM and they both outperform VB and
Bayesian EM. For this observation sequence, the log-likelihood values for segmentation MM are
slightly better  than those for
segmentation EM. But this is not a rule: the results for the other 19
observation sequences show that sometimes segmentation
EM is better, and sometimes the other way around.

The log-likelihood values in Table \ref{ex2logpyx} give again a general
summary measure for comparing the best paths over the studied
methods. Relatively small differences in log-likelihood values can incorporate large pointwise
differences in the respective sequences. For example, the log-likelihood values of the best state sequences for segmentation EM and segmentation MM when $Q=Q_2$ and
$M=50$ are -855.26 and -854.91, respectively. The pointwise difference between the state paths is 163, the transition frequency matrices are given by
\[ \left( \begin{array}{cccc}  44  &  1 &   2 &  12 \\ 0 & 333 &   0 &  16  \\  2 &   0 &   5  &  0 \\  14 &  15  &  0 & 155 \\ \end{array} \right),\quad \quad
 \left( \begin{array}{cccc}  37  &  3 &  0 &  4 \\ 2 & 485 &   0 &  6  \\  0 &   0 &   0  &  0 \\  5 &  6  &  0 & 51 \\ \end{array} \right).
\]
\begin{table}[h!]
\begin{center}
\begin{tabular}{| c | c || c| c | c || c | c | c |}
  \hline
   $Q$ & $M$ & SA$_{max}$  & sEM$_{max}$ & sMM$_{max}$ & SA$_{min}$  & sEM$_{min}$   & sMM$_{min}$  \\
  \hline\hline
$Q_1$  & 600 &    18  &   12  &   11  &   1   &   3  &    6 \\
       &  150 &   12  &    9  &    8  &   7   &   6  &    5 \\
       &   50  &    2  &   15  &   12 &   18  &    0  &    0 \\
       &   10  &    5  &   12  &   11 &   15  &    2  &    0  \\
       &    5  &     5 &    11 &    12  &  15  &    1 &     0 \\
\hline
$Q_2$  & 600 &    17  &    4   &   6  &   3   &   7    &  2  \\
       & 150   &  4   &  10  &   12   & 15   &   3   &   1   \\
       & 50   &   1   &  10   &  14  &  19   &   0   &   1  \\
       & 10   &   0   &   9   &  14  &  17   &   1   &   2   \\
       & 5     &  4   &  16   &   na   & 16   &   4   &   na   \\
\hline
$Q_3$  & 600   & 20  &  11   &  13  &   0   &   4   &   1   \\
       & 150    & 7   &  13   &  10  &  13   &   1   &   2  \\
       & 50   &   1   &  14   &  13  &  18   &   1    &  1   \\
       & 10   &   3   &  10  &   14  &  16   &   2   &  1   \\
       & 5     &  5   &  15  &    na  &  15   &   5   &   na  \\
\hline
\end{tabular}
\end{center}
\caption{\label{n600sum20seqEx2} The counts over the Viterbi path estimates corresponding to 20 different observation sequences showing when the simulated annealing,
segmentation EM and segmentation MM methods reached the maximum and minimum values of log-likelihood in Example 2.
The minimum count shows how many times the respective method performed worst of the three methods according to the log-likelihood value.}
\end{table}
\begin{table}[h!]
\begin{center}
{\small
\begin{tabular}{| c | c || c| c | c | c | c | c | c |}
  \hline
   $Q$ & $M$ & sEM  & sMM & ICM & VB  & B(EM)   & EM & SA  \\
  \hline\hline
$Q_1$ & 600 & -984.19 (34) & \textbf{-984.19} (33) & -984.19 (33) & -988.38 (8) & -988.33 & -1129.93 (35) & -984.19 \\
      & 150 & -964.23 (25) & -964.18 (25) & -964.17 (45) & -975.41 (4) & -974.82 & -1031.09 (35) & \textbf{-963.80} \\
      & 50  & -933.55 (22) & \textbf{-933.47} (23) & -938.50 (45) & -964.38 (2) & -964.28 & -950.33 (35)  & -936.15 \\
      & 10  & \textbf{-854.69} (20) & \textbf{-854.69} (24) & \textbf{-854.69} (46) & -915.08 (1) & -910.94 & -881.74 (35)  & -857.87 \\
      &  5  & \textbf{-839.89} (20) & \textbf{-839.89} (25) & \textbf{-839.89} (46) & -890.10 (9) & -887.64 & -869.43 (35)  & -860.79 \\
\hline
$Q_2$ & 600 & -891.57 (12) & -891.57 (10) & -895.73 (47) & -900.09 (1) & -898.46 & -927.46 (35) & \textbf{-891.53} \\
      & 150 & \textbf{-881.36} (15) & \textbf{-881.36} (17) & -884.86 (47) & -887.66 (1) & -887.26 & -900.19 (35) & -881.37 \\
      & 50  & -855.26 (14) & \textbf{-854.91} (16) & -875.71 (47) & -873.85 (1) & -874.64 & -877.69 (35) & -866.92 \\
      & 10  & -826.88 (20) & \textbf{-822.18} (29) & -857.62 (47) & -864.01 (1) & -858.24 & -859.59 (35) & -841.45 \\
      & 5   & \textbf{-818.95} (24) & na           & -841.14 (47) & -857.46 (1) & na & -857.55 (35) & -841.52 \\
\hline
$Q_3$ & 600 & -938.34 (22) & -938.12 (23) & -938.34 (47) & -954.98 (2) & -950.91 & -1014.17 (35) & \textbf{-938.08} \\
      & 150 & \textbf{-927.62} (17) & \textbf{-927.62} (18) & -935.46 (47) & -936.39 (4) & -936.62 & -958.13 (35) & -929.32 \\
      & 50  & -897.59 (20) & \textbf{-897.14} (18) & -908.92 (47) & -918.83 (2) & -919.05 & -910.96 (35) & -905.39 \\
      & 10  & -840.84 (21) & \textbf{-834.19} (24) & -857.11 (47) & -888.23 (4) & -883.58 & -869.77 (35) & -857.09 \\
      & 5   & \textbf{-826.57} (21) &  na          & -841.30 (47) & -873.73 (1) & na & -862.78 (35) & -849.91 \\
\hline
\end{tabular}}
\end{center}
\caption{\label{ex2logpyx} The best log-likelihood value {$\ln
p(x,\hat{v})$ (calculated with formulas (\ref{p(y)}) and
(\ref{jaana}))} for segmentation EM, segmentation MM, ICM,
variational Bayes, Bayesian EM, standard EM and simulated annealing
methods in Example 2.
 The best result(s) for every set of hyperparameters is presented in bold. The number of different output sequences out of 47 possible is given in the brackets.}
\end{table}
%
%
\section{The role of hyperparameters in Bayesian segmentation} \label{HyperRole}
In this section we will point out some important issues regarding the choice of hyperparameters which might be helpful also for interpretation of segmentation results.
\subsection{Dirichlet priors}
{\paragraph{Uniform Dirichlet priors.} Let us briefly discuss
the case when $\alpha_{lj}=1$ for every $l$ and $j$. Then the rows
of the transition matrix are uniformly distributed and
therefore, the priors with $\alpha_{lj}=1$ are considered to be
non-informative, which corresponds to not assuming
anything of the transition matrix. In other words, all transition matrices are equiprobable
and the expected values of all entries in the transition matrix are ${1\over K}$.
This might suggest that the same holds in the sequence space and no particular path
structure (like sequences with long blocks or rapid changes) is preferred. But this is not the case -- with uniform Dirichlet priors the state
sequences are far from being equiprobable and the ones having long
blocks are preferred. The following proposition proves that sequences with maximum
prior weight are the constant ones.
\begin{proposition}\label{py} Let $\alpha_{lj}=1$ for every $l,j$. Then
$$\arg\max_{y}p(y)=\{(i,\ldots,i),\quad i=\arg\max_i p_{0i}\}.$$
\end{proposition}
\begin{proof} When $\alpha_{lj}=1$, then $p(y)$ is according to (\ref{p(y)}) for any sequence $y$ given by
$$p(y)=p_{0y_1}[\Gamma(K)]^K\prod_l {\prod_j \Gamma(1+n_{lj}(y))\over \Gamma(K+n_l(y))}=p_{0y_1}[\Gamma(K)]^K \prod_l {\prod_j n_{lj}(y)!\over (n_l(y)+K-1)!} .$$
For the proof it suffices to show that any constant sequence maximizes the product term in the expression above.
Fix $y$ and denote $n_{lj}:=n_{lj}(y)$. Since for every $l=1,\ldots,K$, $\sum_j n_{lj}=n_l$ and $\sum_l n_l=n-1$, the following inequality holds for any integer $a>0$:
$\prod_{l=1}^K(n_l+a) \ge a^{K-1}(n-1+a)$, where the equality holds only if $n_l=n-1$ for some $l=1,\ldots,K$. Therefore,
\[ \prod_{l=1}^K{\prod_{j=1}^K n_{lj}!\over (n_l+K-1)!} \leq \prod_{l=1}^K{n_l!\over (n_l+K-1)!} \le \Big({1\over {2\cdot 3\cdots (K-1)})}\Big)^{(K-1)}\Big(n(n+1)\cdots (n+K-2)\Big)^{-1}.\]
The inequality is strict if $\exists j,k$ such that $n_j>0$ and $n_k>0$. Thus, the upper bound is reached only if $n_{l}=n-1$ for some $l=1,\ldots,K$, that is we have a constant state sequence.
\end{proof}
\noindent To summarize: assuming nothing about the transition
matrix is not equivalent to not assuming anything about the state sequences. On the
contrary, equiprobable paths correspond to the fixed
transition matrix with all entries equal to ${1\over K}$, which is
a very specific and strong assumption about the transition matrix.
\paragraph{The role of precision parameter.} Recall the parametrization $\alpha_{lj}=Mq_{lj}$.
When $q_{lj}=1/K$ for every $l$ and $j$, then the precision parameter
$M$ can be considered as a regularization parameter in the optimization problem
\begin{equation}\label{optln}
\max_y\Big(\ln p(x|y)+\ln p_M(y)\Big),
\end{equation}
where the subscript $M$ denotes the dependence on $M$. Increasing
$M$ corresponds to reducing the role of $\ln p(y)$, thus the limit
case $M \to \infty$ corresponds to $\ln p(y)=const$ (all paths are
equiprobable). Therefore, when $M \to \infty$,  (\ref{optln})
reduces to $\max_y p(x|y)$. The case with $M=K$ corresponds to the
case of uniform Dirichlet priors with $\alpha_{lj}=1$, and in this
case the role of $\ln p(y)$ in (\ref{optln}) is to make the output
sequences more constant. Thus, when $q_{lj}=1/K$, then decreasing $M$ means changing the
sequence prior $p_M(y)$ so that the sequences with large blocks will have
more weight.

In this article, we also consider $Q$-matrices, where the
entries on the main diagonal  have larger values than the
off-diagonal elements. With such $Q$, for every $M$ the sequence
prior $p_M(y)$ puts more weight on the  sequences with big blocks
and the most probable sequences are constant ones. However, this
behavior is even more pronounced for small $M$. Indeed, if $M\to
\infty$, then for every  $y$, $p_M(y)\to
p_{0y_1}\prod_{l,j}q_{lj}^{n_{lj}(y)}:=p_{\infty}(y)$. It is easy to
see that for constant sequences the convergence is monotone. For
example, if $y=1,\ldots,1$, then as $M\to \infty$, it holds that
$p_M(y)\searrow p_{01}q_{11}^{n-1}=p_{\infty}(y)$. Thus, since the
entries on the main diagonal  have larger values than the
off-diagonal elements, the limit measure $p_{\infty}$ puts more
weight on sequences with large blocks. But due to the monotone
convergence, we see that for smaller $M$ the measure $p_M(y)$
concentrates on such sequences even more.
{\subsection{Clustering under normal emissions with NIX priors} \label{clusters} To understand the
role of emission hyperparameters, it is instructive to consider the optimization problem $\max_y p(x|y)$.  In the Bayesian
HMM setup this corresponds to the limit case $M\to\infty$ when $q_{lj}=1/K$ $\forall l,j$, thus $p_M(y)=const$. Since $p_M(y)$ is not involved in
segmentation anymore, the whole temporal structure of the model is dropped and it is more correct to refer to the problem as {\it clustering}.
We will show that the nature of the clustering problem and its solutions depend heavily
on the hyperparameters. It turns out that under NIX priors, the
family of possible clustering problems is large, including many
familiar $k$-means related problems. We will briefly discuss some of
them. Typically, `standard' problems are obtained when the hyperparameters $\nu_0$, $\kappa_0$ and $\tau^2_0$ approach their
extreme values, that is $0$ or $\infty$. The details about the formulae are given in the Appendix.
\paragraph{I. The case $\nu_0\to \infty$.} When $\nu_0\to \infty$, then the problem $\max_y p(x|y)$ approaches for given $\tau^2_0>0$ the following clustering problem: find clusters
$S_1,\ldots,S_K$ that minimize
\begin{equation}\label{clust}
\sum_{k=1}^K \Big[\sum_{t\in S_k}
(x_t-\bar{x}_k)^2+{\kappa_0 m_k\over \kappa_0+m_k}
(\bar{x}_k-\xi_k)^2+\tau_0^2 \ln (\kappa_0+m_k)\Big],\end{equation}
which is equivalent to minimizing
\begin{equation} \label{clustequiv}
\sum_{k=1}^K \min_{\mu_k\in {\cal X}}\Big[\sum_{t\in
S_k}(x_t-\mu_k)^2+\kappa_0 (\mu_k-\xi_k)^2+\tau_0^2 \ln
(\kappa_0+m_k)\Big],  \end{equation}
where $m_k=|S_k|$ (see the Appendix). The first term in (\ref{clustequiv}) corresponds to the  sum of least squares,
the second term tries to form clusters around $\xi_k$ and the third term tries to build clusters
of unequal size. Thus, if $\tau_0^2$ is very big, then one
cluster gets very big and the others are empty or very small. For small $\tau^2_0$, the influence of the third term is small.
When $\kappa_0 \to 0$, the second term disappears. This corresponds to the case where the variance of $\mu_k$ is infinite (uninformative prior for $\mu_k$).
The case with $\kappa_0 \to 0$ and $\tau_0^2 \to 0$ corresponds to the classical $k$-means optimization problem.

The segmentation MM algorithm acts in the case $\nu_0 \to \infty$
for any $\tau_0^2$ as follows: given clusters
$S^{(i)}_1,\ldots,S^{(i)}_K$, find the corresponding cluster centres
$$\mu^{(i)}_k=(m^{(i)}_k\bar{x}^{(i)}_k+\kappa_0\xi_k ) / (\kappa_0+m^{(i)}_k ).$$
Given these centres, find new clusters corresponding to the Voronoi partition:
$$S_k^{(i+1)}=\{x_t:|x_t-\mu_k^{(i)}|=\min_l|x_t-\mu_l^{(i)}|\}.$$
In the case of segmentation EM algorithm the cluster centres $\mu^{(i)}_k$ are calculated in the same way as for segmentation MM, but the
clustering rule is different:
\begin{equation}\label{tyk}
S_k^{(i+1)}=\left\{x_t:\big(x_t-\mu_k^{(i)}\big)^2+{\tau_0^2 \over
m^{(i)}_k+\kappa_0}=\min_l \left[ \big(x_t-\mu_l^{(i)}\big)^2+{\tau_0^2
\over m^{(i)}_l+\kappa_0}\right]\right\}.\end{equation}
The term $\tau_0^2 /(m^{(i)}_l+\kappa_0)$ in (\ref{tyk}) affects cluster size. When $\tau^2_0$ is small, then segmentation EM and MM give the same result,
but when $\tau^2_0$ increases, then segmentation EM tends to produce clusters of unequal size, whereas segmentation MM remains unaffected by $\tau^2_0$.
When $\kappa_0 \to 0$ and $\tau^2_0\to 0$, then both algorithms converge to the standard Lloyd algorithm.

{When $\kappa_0\to \infty$, the clustering problem in (\ref{clust})
reduces to minimizing $\sum_{k=1}^K\sum_{t\in S_k} (x_t-\xi_k)^2$
with the solution given by $ S_k=\{x_t: |x_t-\xi_k|=\min_l|x_t-\xi_l|\}$.
The solution matches fully with intuition, because $\nu_0 \to\infty$ and
$\kappa_0 \to \infty$ corresponds to the case with fixed normal emissions with means $\xi_k$
and variances $\tau_0^2$, thus clustering is trivial.
%
\paragraph{II. The case with finite $\nu_0$.} For a given $\nu_0$, the general optimization problem we have is the following: find clusters $S_1,\ldots,S_k$ minimizing the sum
\begin{equation}\label{big}
-\sum_k \ln \Gamma({\nu_0+m_k\over 2})+{1\over 2}\sum_k \ln
(\kappa_0+m_k)+\sum_k{\nu_0+m_k\over 2}\min_{\mu_k\in {\cal X}} \ln
\big(\nu_0\tau^2_0+ \sum_{t\in S_k}
(x_t-\mu_k)^2+\kappa_0(\mu_k-\xi_k)^2\big).\end{equation}
The first two terms in (\ref{big}) tend to make size of the clusters unequal. This follows from the observation that under the constraint
$\sum_k m_k=n$, the products
$$\prod_{k=1}^K\Gamma({\nu_0+m_k\over 2}),\quad \prod_{k=1}^K(\kappa_0+m_k)^{-{1\over 2}}$$
are maximized when $m_k=n$ holds for some $k$. The smaller $\nu_0$ and $\kappa_0$ are, the bigger is the influence of the first two terms.
When $\kappa_0\to \infty$, the problem of finding optimal clusters $S_1,\ldots,S_K$ reduces to minimizing
$$-\sum_k \ln \Gamma({\nu_0+m_k\over 2})+\sum_k{\nu_0+m_k\over 2} \ln
\big(\nu_0\tau^2_0+ \sum_{t\in S_k} (x_t-\xi_k)^2). $$
The solution to this problem gives bigger clusters than obtained by
minimizing $\sum_{k=1}^K\sum_{t\in S_k} (x_t-\xi_k)^2$. When $\tau_0^2\to \infty$, then the
last term in (\ref{big}) disappears and the problem reduces to finding clusters that minimize the sum of the first two terms in (\ref{big}).
The solution here is one big cluster.

As our examples have shown, clustering under NIX setting is highly dependent on hyperparameters, and the choice of hyperparameters can strongly affect the resulting segmentation.
%
%
%
\section{Similarity of the algorithms studied} \label{disc}
We have seen that out of the five non-stochastic optimization methods
(segmentation EM, segmentation MM, ICM, Bayesian EM and VB), ICM is
clearly most inadequate, because it depends heavily on
initial sequences and gets stuck in local optima. The other four
methods can be divided into two groups, which can be characterized
as segmentation-based methods (segmentation EM and segmentation MM)
and parameter-based methods (Bayesian EM and VB).
We call VB a parameter-based method, because it updates the parameters iteratively and then, with final $h_k$ and
$u_{lj}$, the Viterbi algorithm is applied (see Subsection \ref{A1}). The segmentation EM and MM methods apply the
Viterbi algorithm at each iteration step. Our numerical examples
demonstrate a clear advantage of the segmentation-based methods,
which is also expected, because segmentation EM optimizes the
objective function of interest and segmentation MM behaves very
similarly.

We already observed the pairwise similarity of the
segmentation-based methods and the parameter-based methods in
Examples 1 and 2. In the case that emission distributions are known, the
four algorithms can be further summarized as follows. Comparing
(\ref{solution3EM}) and (\ref {VBu}) shows that both the Bayesian EM
and VB updates can be written as
\[
\ln
p_{lj}^{*(i+1)}=f_1\big(\xi^{(i)}(l,j)+Mq_{lj}\big)-f_2\big(\sum_j\xi^{(i)}(l,j)+M\big),\]
where $p^*_{lj}$ is either $p_{lj}$ (Bayesian EM) or $u_{lj}$ (VB),
and where $f_1(x)=\ln (x-1)$, $f_2(x)=\ln (x-K)$ for Bayesian EM and
$f_1=f_2=\psi$ for VB. Similarly, the transition updates for
segmentation MM (\ref{MM}) and segmentation EM (\ref{uij}) can be
written as
\[ \ln p_{lj}^{*(i+1)}=f_1\big(n_{lj}(y^{(i)})+Mq_{lj}\big)-f_2\big(n_{l}(y^{(i)})+M\big),\]
where $f_1(x)=\ln (x-1)$, $f_2(x)=\ln (x-K)$ for segmentation MM and $f_1=f_2=\psi$ for segmentation EM. Thus, the four methods can be
characterized by two parameters: the {\it function} parameter ($\ln$  vs $\psi$) and the {\it counts} parameter
(direct counts $n_{lj}(s)$ versus averaged counts $\xi(l,j)$):
\begin{center}\begin{tabular}{|l|c|c|}
  \hline
  Counts/Function & $\ln$ & $\psi$ \\\hline
  Direct ($n_{lj}$) & sMM & sEM \\\hline
  Averaged ($\xi(l,j)$) & B(EM) & VB \\
  \hline
\end{tabular}
\end{center}
The results of Examples 1 and 2 show that the
difference in functions does not influence the algorithm as much
as the difference in counts, because the methods behave similarly row-wise. The examples also show that in terms of maximizing the main study criterion, that is the
posterior likelihood, the methods using direct counts outperform the methods that use averaged counts.
We have noticed that the methods using $\ln$-function give
sometimes slightly larger posterior probability than the ones using $\psi$, and this is a matter for future research.
%
%
\section{Conclusions and further research}
%
The paper is mainly devoted to studying non-stochastic
algorithms for finding the Viterbi path in Bayesian hidden Markov
models. The performance of the segmentation EM method introduced in the article
has been compared with other well-known non-stochastic methods (segmentation MM, iterative conditional mode, variational Bayes, Bayesian EM) as well as with the simulated annealing approach.

The segmentation EM method that optimizes the correct objective
function mostly outperforms the other studied methods, often also
the simulated annealing method. It should be noted that the
possibility to apply the segmentation EM algorithm should not be
taken for granted for any model. For many models the EM algorithm
can be written down easily theoretically, but the maximization
and/or expectation step can be impractically complicated to perform.
One example of such a model is the hidden Markov model with infinite
state space (hierarchical Dirichlet processes), where the E-step
involves intractable integrals. In our setup {with Dirichlet prior
distributions and emissions from the exponential family} the E-step
involves well-known digamma functions and the M-step reduces to the
Viterbi algorithm, therefore the segmentation EM is easily
applicable.

Our study demonstrates that when the main goal of inference is segmentation,
then the Bayesian approach should be used. The Bayesian setup enables to
concentrate directly on segmentation and skip the parameter estimation step.

It is a little surprising that the segmentation MM method behaves in
our examples as well as the segmentation EM algorithm, since the
performance of the same method in the context of parameter
estimation (known then as Viterbi training) is often notoriously
bad. The similarity of the segmentation EM and MM methods is shortly
discussed in Section 5, but a good performance of segmentation MM
needs further investigation. For practitioners we advise to be
careful with the segmentation MM method, because it does not
optimize the right criterion function as segmentation EM does.

The segmentation EM and MM methods are sensitive with respect to
initial sequences, therefore the choice of initial sequences is
crucial. Since for both algorithms it is actually the empirical
transition matrix of the initial sequence that is the input to the
algorithm, initial sequences should be chosen so that the
corresponding empirical transition matrices are different and
somehow cover the search space.

The article brings out the important role of hyperparameters in the Bayesian context,
different issues regarding this topic are thoroughly discussed in Section 4. Our results
demonstrate that hyperparameters determine largely the nature of the segmentation problem
and the properties of the solution, they also control the influence of data. The simulation
examples show that even a small change in some of the hyperparameters can change
the problem drastically. This is obviously a disappointment for practitioners because the idea
of Bayesian approach is to get rid off the choice of parameters, and now it turns out that the
hyperparameters should be chosen equally carefully. It seems to
us that the role of hyperparameters is overlooked in the
literature, at least in the segmentation context.

The concluded research opens several interesting directions for future studies.
As pointed out in Introduction, a common alternative to the Viterbi path in practice is
the PMAP path, which is the state path estimate that minimizes the expected number of
classification errors. For given parameters (known or estimated), the PMAP path can be found with the well-known forward-backward algorithm.
How to find the PMAP path in the Bayesian setup is an open and challenging question, since there is no obvious analogue to
the segmentation EM or MM algorithm in this case.

Another appealing research question is about incorporating
inhomogeneity to the model. In the Bayesian setup inhomogeneity
means the change of priors from time to time. In the case of known
change points and independent priors the situation reduces to
cutting the whole model into independent submodels. However, in
general and thus even in the Bayesian setup it might instead be
preferable to consider the model where the change points are not
exactly known. Suppose there are a few possible transition matrices
$\{\mathbb{P}_i\}$ and the underlying chain $Y$ is inhomogeneous
driven by one of these matrices at a time. However, we do not know
{\it a priori} which matrix drives the transition at a given time
$t$. An elegant way for incorporating such kind of variability and
information into the model is the so-called triplet Markov models
(TMMs) introduced by Piecynski \cite{P07}. In TMMs, instead of a
Markov chain $Y$ a bivariate Markov chain $(Y,U)$ is considered,
where the additional component $U$ allows a change of the transition
matrix. Since $Y$ is not a Markov chain, the pair  $(X,Y)$ is not an
HMM anymore, and therefore it is not obvious how to find the Viterbi
path in this model. A closer inspection indicates that segmentation
EM might still be applicable, at least under some additional
assumptions. A further step would be to consider a hierarchical
model where the Dirichlet hyperparameters, say $\alpha$, are modeled
in the way described, that is $(\alpha,U)$ is a bivariate Markov
chain. This incorporates both the approach with Dirichlet transition
priors and the approach with variable change points.

Since segmentation in the Bayesian setup heavily depends on
hyperparameters, it would be tempting to put additional priors on
hyperparameters. Such models are sometimes called hierarchical.
Another example of a hierarchical model is hierachical Dirichlet
processes (see \cite{BayesNonparam}, Ch. 5), where the number of
hidden states is not fixed any more. Such models are complicated and
how to design non-stochastic segmentation algorithms in this case is
a very interesting research area.
%
%
\section{Appendix}
\subsection{General formulae for the  segmentation methods studied} \label{A1}
Due to our independence assumption, all emission and transition parameters can be estimated separately. In the formulae of this section we use the same notation for the random parameters $p_{l,j}$, $\mu_k$ and $\sigma_k^2$, $k,l,j\in \{1,\ldots,K\}$, and the corresponding estimates. The exact meaning can be understood from the context.
\paragraph{Segmentation MM.}
\noindent In the case of Dirichlet priors the matrix $\theta_{tr}^{(i+1)}$ can be found row-wise, the $l$-th
row is the posterior mode:
\begin{equation}\label{MM}
p^{(i+1)}_{lj}={\alpha_{lj}+n_{lj}(y^{(i)})-1\over
\alpha_l+n_l(y^{(i)})-K}.\end{equation}
Emission parameters can be updated independently:
\[\theta_{em}^{k,(i+1)}=\arg\max_{\theta^k_{em}}p(\theta^k_{em}|x_{S_k})=
\arg\max_{\theta^k_{em}}\Big[\sum_{t: y^{(i)}_t=k} \ln
f_k(x_t|\theta^{k}_{em})+\ln \pi^k_{em}(\theta^{k}_{em})\Big] ,
\quad k=1,\ldots,K,\]
where $x_{S_k}$ is the subsample of
$x$ corresponding to  state $k$ in $y^{(i)}$. Formally, for every
sequence $y\in S^n$ define $S_k(y)=\{t\in \{1,\dots,n\}:
y_t=k\}$, then $x_{S_k}=\{x_t: t\in S_k\}$.
\paragraph{Bayesian EM.}
\noindent The emission updates are given by
\begin{equation} \label{tarn}
 \theta_{em}^{k,(i+1)}=\arg\max_{\theta^k_{em}}\Big[\sum_t \ln f_k(x_t|\theta^{k}_{em})\gamma_t^{(i)}(k)+\ln
\pi^k_{em}(\theta^{k}_{em})\Big],\quad k=1,\ldots,K, \end{equation}
where
\begin{equation}\label{defgamma}
\gamma_t^{(i)}(k):=P(Y_t=k|X=x,\theta^{(i)})=\sum_{y:
y_t=k}p(y|\theta^{(i)},x).\end{equation}
In the case of Dirichlet priors the transition updates are given by
\begin{align}\label{BEMsolution}
{p}^{(i+1)}_{lj}&={\xi^{(i)}(l,j)+ (\alpha_{lj}-1)\over \sum_j
\xi^{(i)}(l,j)+(\alpha_l-K)},\quad \mbox{where} \quad
\xi^{(i)}(l,j):=\sum_{t=1}^{n-1}P(Y_{t}=l,Y_{t+1}=j|x,\theta^{(i)}).\end{align}
Since one of the studied methods (ICM) starts with an initial
sequence, in order the comparison to be fair, we let all the other
methods to start with a sequence as well. Therefore, for a given initial sequence
$y^{(0)}$, define
\begin{equation}\label{gamma0}
\gamma_t^{(0)}(k):=I_{k}(y_t^{(0)}),\quad
\xi^{(0)}(l,j):=n_{lj}(y^{(0)}).\end{equation}
\paragraph{Variational Bayes approach.}
\noindent Let  us have a closer look at the measure $\Qy^{(i+1)}(y)$. We are going to show that there exists an HMM
$(Z,X)$ such that for every sequence $y$, $\Qy^{(i+1)}(y)={P}(Z=y|X=x)$.
By definition,
$$\Qy^{(i+1)}(y)\propto \exp\Big[\int   \ln p(\theta,y|x)
\Qt^{(i+1)}(d\theta)\Big].$$
Apply the notation from \eqref{seg-EM-V} in the current case:
$$u^{(i+1)}_{lj}=\exp[\int \ln p_{lj}(\theta_{tr}) \Qt^{(i+1)}(d\theta)],\quad
h^{(i+1)}_k(x_t)=\exp[\int \ln f_k(x_t|\theta_{em}^k)\Qt^{(i+1)}(d\theta)].$$
Since $\ln p(\theta,y|x)=\ln \pi(\theta)+\ln p(y,x|\theta)-\ln p(x)$, we obtain
%
\begin{align} \notag
&\int \ln p(\theta,y|x) \Qt^{(i+1)}(d\theta)=\int \ln
\pi(\theta)\Qt^{(i+1)}(d\theta)-\ln
p(x) + \int \ln p(y,x|\theta)\Qt^{(i+1)}(d\theta) \\ \notag
&=c(\Qt^{(i+1)},x)+ \ln p_{0 y_1}+\sum_{lj}n_{lj}(y)\ln
u^{(i+1)}_{lj}+\sum_{k=1}^K
\sum_{t: y_t=k} \ln h^{(i+1)}_k(x_t)\\ \label{Zhmm}
&=c(\Qt^{(i+1)},x)+ \ln p_{0 y_1}+\sum_{lj}n_{lj}(y)\ln
\tilde{u}^{(i+1)}_{lj}+\sum_{k=1}^K \sum_{t: y_t=k} \ln \tilde{h}^{(i+1)}_k(x_t),
\end{align}
where $\tilde{u}_{lj}$ is the normalized quantity, $
\tilde{u}_{lj}:={u_{lj}\over \sum_j u_{lj}}$, and
$\tilde{h}_k(x_t):=(\sum_j u_{kj})h_k(x_t)$, if $t\leq n-1$,
$\tilde{h}_k(x_n)=h_k(x_n)$. Let now $(Z,X)$ be an HMM, where $Z$ is
the underlying Markov chain with transition matrix $(\tilde{u}_{lj})$
and emission densities are given by $\tilde{h}_k$. From (\ref{Zhmm})
it follows that $\Qy^{(i+1)}(y)\propto P(Z=y|X=x)$. Since $\Qy^{(i+1)}$ and
$P(Z\in \cdot |X=x)$ are both probability measures, it follows that they are equal.
To stress the dependence on iterations, we will denote $\Qy^{(i+1)}(y)=P^{(i+1)}(Z=y|X=x)$.

Let us now calculate  $\Qt$. Let $\gamma_t^{(i)}(k)$ denote the marginal of $\Qy^{(i)}(y)$,
 $$\gamma_t^{(i)}(k):=P^{(i)}(Z_t=k|X=x)=\sum_{y:y_t=k}\Qy^{(i)}(y).$$
Observe that
\[\sum_{y}\ln p(y,x|\theta)\Qy^{(i)}(y)=C_1 +\sum_{l,j}\ln
p_{lj}(\theta_{tr})\big(\sum_y
n_{lj}(y)\Qy^{(i)}(y)\big)+\sum_{t=1}^n\sum_{k=1}^K\ln
f_k(x_t|\theta^{k}_{em})\gamma^{(i)}_t(k),\]
where $C_1:=\sum_{k}(\ln p_{0k})\gamma_1^{(i)}(k)$.
The sum $\sum_y n_{lj}(y)\Qy^{(i)}(y)$ is the expected number of
transitions from $l$ to $j$, so that using the equality
$\Qy^{(i)}(y)=P^{(i)}(Z=y|X=x)$, we have
\[\sum_y n_{lj}(y)\Qy^{(i)}(y)=\sum_{t=1}^{n-1}P^{(i)}(Z_t=l,Z_{t+1}=j|X=x)=:\xi^{(i)}(l,j). \]
Therefore,
\begin{equation}\label{VBupdate}
 \ln \Qt^{(i+1)}(\theta)=C+\ln \pi_{tr}(\theta_{tr})+\ln \pi_{em}(\theta_{em})+
\sum_{l,j}\xi^{(i)}(l,j)\ln p_{lj}(\theta_{tr})+\sum_{t=1}^n\sum_{k=1}^K \ln f_k(x_t|\theta^{k}_{em})\gamma^{(i)}_t(k).\end{equation}
From (\ref{VBupdate}) we can see that under $\Qt^{(i+1)}$ the
parameters $\theta_{tr},\theta^{1}_{em},\ldots,\theta^{K}_{em}$ are
still independent and can therefore be updated separately. In the
case of Dirichlet transition priors the rows are independent as well.
The transition update for the $l$-th row and the emission update for the $k$-th component are given by
\[\Qt^{(i+1)}(p_{l1},\ldots,p_{lK})\propto \prod_{j=1}^K
p_{lj}^{\alpha_{lj}-1+\xi^{(i)}(l,j)}, \quad
\Qt^{(i+1)}(\theta_{em}^k)\propto \pi(\theta_{em}^k)\prod_{t=1}^n \big( f_k(x_t|\theta_{em}^k)\big)^{\gamma_t^{(i)}(k)}. \]
The whole VB approach is applicable since $\xi^{(i)}(l,j)$
 and $\gamma_t^{(i)}(k)$ can be found
by the standard forward-backward formulae using
$\tilde{u}_{lj}^{(i)}$ and $\tilde{h}^{(i)}_k$. Actually, it is not
difficult to see that in these formulae the original $u^{(i)}_{lj}$
and $h^{(i)}_k$ can be used instead of the standardized ones.
To summarize, in our setup the VB approach
yields the following algorithm for calculating $\hat{v}_{\rm{VB}}$. For a given initial sequence $y^{(0)}$, find vector
$\gamma_t^{(0)}$ and matrix $\xi^{(0)}$ as in (\ref{gamma0}). Given $\gamma_t^{(i)}$ and $\xi^{(i)}$, update
$u^{(i+1)}_{lj}$ and $h_k^{(i+1)}$ as follows:
\begin{equation}\label{VBu}
u^{(i+1)}_{lj}=\exp[\psi(\alpha_{lj}+\xi^{(i)}(l,j))-\psi(\alpha_{l}+\xi^{(i)}(l))],\quad
\text{where}\quad \xi^{(i)}(l):=\sum_j \xi^{(i)}(l,j),\end{equation}
\[ h^{(i+1)}_k(x_t)=\exp[\int \ln f_k(x_t|\theta_{em}^k)\Qt^{(i+1)}(d\theta)],\quad
\text{where}\quad \Qt^{(i+1)}(\theta_{em}^k)\propto
\pi(\theta_{em}^k)\prod_{t=1}^n \big(f_k(x_t|\theta_{em}^k)\big)^{\gamma_t^{(i)}(k)}.\]
With these parameters, $\xi^{(i+1)}$  and $\gamma_t^{(i+1)}$ can be calculated with the usual forward-backward procedure for
HMM. Then update $u^{(i+2)}_{lj}$ and $h_k^{(i+2)}$ and so on. After
the convergence, say after $m$ steps, apply the Viterbi algorithm
with transitions $(u_{ij}^{(m)})$ and emission densities
$h^{(m)}_k$. The obtained path  maximizes $\Qy^{(m)}(y)$ over
all the paths, so it is $\hat{v}_{\rm{VB}}$.
\paragraph{Simulated annealing.} Because of independence of the emission and transition
parameters, it holds even for $\beta>1$ that
$p_{\beta}(\theta|y,x)=p_{\beta}(\theta_{tr}|y)p_{\beta}(\theta_{em}|y,x)$, thus
the transition and emission parameters can  be sampled
separately. When the rows of a transition matrix have independent
Dirichlet priors, the $l$-th row can be generated from the
Dirichlet distribution with parameters
$\beta(n_{lk}(s)+\alpha_{lk})+1-\beta$, $k=1,\ldots,K$. For given
$\theta$, sampling from $p(y|\theta,x)$ can be performed in various
ways: we use so-called {\it Markovian Backward Sampling} (Algorithm 6.1.1 in \cite{HMMbook}). To sample from
$p_{\beta}(y|\theta,x)$, note that
$$p(x,y|\theta)^{\beta}={p^{\beta}_{0y_1}\over
\sum_{j}p^{\beta}_{0,j}}\prod_{t=2}^{n}{\tilde p}_{y_{t-1} y_{t}}\tilde{f}_{y_t}(x_t),$$
where $\tilde{p}_{ij}:={p^{\beta}_{ij} / \sum_{j}p^{\beta}_{ij}}$,
and
$\tilde{f}_k(x_t):=\big(\sum_{j}p^{\beta}_{ij}\big)f^{\beta}_k(x_t)$, $t=1,\ldots,n-1$, $\tilde{f}_k(x_n):=\big(\sum_{j}p^{\beta}_{0j}\big)f^{\beta}_k(x_n).$
Although the functions $\tilde f_k$ are not densities, one can still
use Markovian Backward Sampling.
\subsection{Non-stochastic segmentation algorithms for the Dirichlet-NIX case}\label{A3}
Suppose the emission distribution corresponding to state $k$ is
$\mathcal N(\mu_k,\sigma_k^2)$, where prior distributions for $\mu_k$ and $\sigma_k^2$ are given by a normal and scaled
inverse-chi-square distribution, respectively:
$$\mu_k |\sigma_k^2\sim \mathcal N \big(\xi_k,{\sigma_k^2\over \kappa_0}\big),\quad
\sigma^2_k\sim {\rm Inv-}\chi^2(\nu_{0},\tau^2_{0}).$$
Here $\kappa_0$, $\nu_{0}$ and $\tau^2_{0}$ are hyperparameters that might depend on $k$, but  in our example we assume they
are  equal. Recall the density of ${\rm Inv-}\chi^2(\nu,\tau^2)$:
$$f(x;\nu,\tau^2)={(\tau^2\nu /2)^{\nu/2}\over
\Gamma({\nu/2})}x^{-(1 + \nu /2)}\exp[-{\nu \tau^2\over 2x}].$$  If $X\sim {\rm Inv-}\chi^2(\nu,\tau^2)$, then
$$EX={\tau^2\nu \over \nu-2},\quad {\rm Var}(X)={2\tau^4\nu^2 \over (\nu-2)^2(\nu-4)},\quad  E(\ln X)=\ln \big({\nu \tau^2\over 2}\big)-\psi \big({\nu \over
2}\big),\quad E X^{-1}=\tau^{-2},$$
and the mode of the distribution is given by $\nu \tau^2/(\nu+2)$.
Therefore, if $\nu_0$ and $\kappa_0$ are both very large, then $\sigma_k^2\approx \tau_0^2$ and $\mu_k\approx \xi_k$, and we get back to the first example. If $\nu_0$ is very large, then
$\sigma^2_k\approx \tau_0^2$, so that emission variances are
$\tau_0^2$, but  the variance of the mean is approximately $\tau_0^2 /\kappa_0$.

Since emission and transition parameters are
independent, the transition parameters can be updated as previously, that is as described in Section \ref{A1}.
Because the emission components $(\theta^1_{em},\ldots,\theta^K_{em})$ are independent under prior
and posterior, it holds that $p(\theta_{em}|y,x)=\prod_k p(\theta^k_{em}|x_{S_k})$, where $x_{S_k}$ is the subsample of $x$ along $y$ corresponding to state $k$.
Let $m_k(y)$ be the size of $x_{S_k}$. Let $\bar{x}_k$ and $s^2_k$ be the mean and variance of $x_{S_k}$.
Since NIX-priors are conjugate, for any state $k$
the posterior parameters $\kappa_{k}$, $\nu_{k}$, $\mu_{k}$ and $\tau_{k}^2$ can be calculated as follows:
\begin{align}\label{k-nid}
\kappa_{k}&=\kappa_0 + m_k\, , \quad \quad \nu_{k} =\nu_0+m_k,
\\\label{k-nid1}
\mu_{k} & = \frac{\kappa_0}{\kappa_0+m_k}\xi_k +
\frac{m_k}{\kappa_0+m_k}\bar{x}_k, \\\label{k-nid2} \nu_{k}
\tau_{k}^2& =\nu_0\tau_0^2 + (m_k-1)s^2_k + \frac{\kappa_0 m_k
}{\kappa_0 +m_k}(\bar{x}_k-\xi_k)^2, \quad
s_k^2=\frac{1}{m_k-1}\sum_{t\in S_k} (x_t-\bar{x}_k)^2,
\end{align}
see \cite{Murphy}. We also need to calculate for every path $y$ the
joint probability $p(x,y)=p(y)p(x|y)$. Due to the independence of
transition and emission parameters, $p(y)$ is still as in
(\ref{p(y)}) and $p(x|y)$ depends on emission parameters, only.
According to the formula for the marginal likelihood (see, e.g.
\cite{Murphy}) we obtain
\begin{align}\label{jaana}
p(x|y)&=\prod_{k=1}^K\int \prod_{t\in
S_k}f_k(x_t|\theta^k_{em})\pi(\theta^k_{em})d
\theta^k_{em}=\prod_{k=1}^K{\Gamma({\nu_{k}\over 2})\over \Gamma({\nu_{0}\over
2})}\sqrt{{\kappa_{0}\over \kappa_{k}}}{(\nu_0\tau_0^2)^{\nu_0\over 2}\over (\nu_{k}\tau^2_{k})^{\nu_{k}\over 2}}\pi^{-{m_k\over 2}}.\end{align}
We will now give a more detailed description of the non-stochastic algorithms for Example 2.
\paragraph{Bayesian EM.} Start with initial state sequence $y^{(0)}$. With this sequence, find for any state $k$ the parameters
$\kappa_{k}$, $\mu_{k}$, $\nu_{k},\tau^2_{k}$ as
defined in  (\ref{k-nid}), (\ref{k-nid1}),
(\ref{k-nid2}) and calculate the posterior modes, that is update
\begin{align*}
p_{lj}^{(1)}&={n_{lj}(y^{(0)})+(\alpha_{lj}-1)\over {\sum_j n_{lj}(y^{(0)})+(\alpha_{l}-K)}},\quad
\mu^{(1)}_k=\mu_{k},\quad (\sigma_k^2)^{(1)}={\nu_{k}\tau_{k}^2\over \nu_{k}+2},\quad k=1\ldots,K.
\end{align*}
With these parameters calculate the vectors $\gamma_t^{(1)}$ and
matrix $(\xi^{(1)}(l,j))$ as in  (\ref{defgamma}) and
(\ref{BEMsolution}) using the forward-backward formulae. Given
$\gamma_t^{(i)}$ and $\xi^{(i)}(l,j)$, the transition parameters are
updated according to (\ref{BEMsolution}).
The emission updates are given by (\ref{tarn}). Let
us calculate $\theta_{em}^{k,(i+1)}$ for the NIX-model. Suppress $k$
from the notation and observe that
$\theta^{(i+1)}_{em}=(\mu^{(i+1)},(\sigma^2)^{(i+1)})$ maximizes the
following function over $\mu$ and $\sigma^2$:
\begin{align*}
&\sum_t
\ln f(x_t|\mu,\sigma^2)\gamma^{(i)}_t+\ln \pi(\mu|\sigma^2)+\ln \pi(\sigma^2)=\\
&{\rm const}-{1\over 2}\Big[ (\ln \sigma^2) \big(\sum_t \gamma^{(i)}_t + (\nu_0 +3) \big)+  {1\over \sigma^2} \big(\sum_t {(x_t-\mu)^2}\gamma^{(i)}_t+{(\mu-\xi)^2\kappa_0}+{\nu_0 \tau_0^2}\big)\Big].
\end{align*}
The solutions $\mu_k^{(i+1)}$ and $(\sigma_k^{2})^{(i+1)}$ are given
by:
\[\mu_k^{(i+1)}={\sum_t x_t \gamma_t^{(i)}(k)+\xi_k\kappa_0\over \sum_t
\gamma^{(i)}_t(k) + \kappa_0},\quad (\sigma_k^{2})^{(i+1)}={\nu_0\tau_0^2+\sum_t(x_t-\mu_k^{(i+1)})^2\gamma^{(i)}_t(k)+(\mu_k^{(i+1)}-\xi_k)^2\kappa_0\over \sum_t \gamma^{(i)}_t(k) + \nu_0+3}.\]
With $\kappa_0\to 0$ (non-informative prior), $\mu_k^{(i+1)}$ is the
same as in the standard EM algorithm. Using the updated parameters, calculate $\gamma_t^{(i+1)}$ and $\xi^{(i+1)}(l,j)$.
Keep updating until the change in the log-likelihood is below the
stopping criterion.
%
%
%
\paragraph{Segmentation EM.}
Given sequence $y^{(i)}$, calculate for every state $k$ the parameters
$\kappa_{k}^{(i)}$, $\mu^{(i)}_{k}$, $\nu^{(i)}_{k}$ and $(\tau_{k}^2)^{(i)}$ using formulae
(\ref{k-nid}), (\ref{k-nid1}) and (\ref{k-nid2}). With these parameters,  calculate $h^{(i+1)}_k(x_t)$ as follows:
\begin{align}\label{hht}
\ln h^{(i+1)}_k(x_t)=-\frac{1}{2}\ln\Big(2\pi (\tau_{k}^2)^{(i)}\Big)-
\frac{1}{2}\left[ \ln\left( \frac{\nu^{(i)}_{k}}{2}
\right)-\psi\left(\frac{\nu^{(i)}_{k}}{2}\right)  \right]
-\frac{x_t^2}{2\left(\tau^{2}_{k}\right)^{(i)}}+x_t\frac{\mu^{(i)}_{k}}{\left(\tau^{2}_{k}\right)^{(i)}}-
\frac{1}{2}\left[\frac{1}{\kappa^{(i)}_{k}}+\left(\frac{\mu^{(i)}_{k}}{\tau_{k}^{(i)}}\right)^2
\right]. \end{align}
Compute the matrix $(u_{lj}^{(i+1)})$, where $\ln u^{(i+1)}_{lj}=\psi(\alpha_{lj}+n_{lj}(y^{(i)}))-\psi(\alpha_{l}+n_{l}(y^{(i)}))$.
To find $y^{(i+1)}$, apply the Viterbi algorithm with $u_{lj}^{(i+1)}$ and $h^{(i+1)}_k(x_t)$.
Keep doing so until no changes occur in the path estimate.
%
%
%
\paragraph{Segmentation MM.}
Given $y^{(i)}$, calculate $\mu^{(i)}_{k}$, $\nu^{(i)}_{k}$ and $(\tau_{k}^2)^{(i)}$ using
formulae (\ref{k-nid}), (\ref{k-nid1}) and (\ref{k-nid2}) and update
the posterior modes as follows:
\begin{align*}
p_{lj}^{(i+1)}&={n_{lj}(y^{(i)})+(\alpha_{lj}-1)\over {\sum_j n_{lj}(y^{(i)})+(\alpha_{l}-K)}},\quad
\mu^{(i+1)}_k=\mu^{(i)}_{k},\quad
(\sigma_k^2)^{(i+1)}={\nu^{(i)}_{k}\over \nu^{(i)}_{k}+2}(\tau_{k}^2)^{(i)}.
\end{align*}
With these parameters find $y^{(i+1)}$ by the Viterbi algorithm. Keep
doing so until no changes occur in the estimated state path.
\paragraph{VB algorithm.} Given an initial state sequence $y^{(0)}$, find  $h^{(1)}_k(x_t)$ and $u^{(1)}_{lj}$ as in
the segmentation EM algorithm. With these parameters, calculate $\gamma^{(1)}_t$ and $\xi^{(1)}(l,j)$ using the forward-backward
formulae. Given the matrix $\big(\xi^{(i)}(l,j)\big)$, update the
matrix  $(u_{lj}^{(i+1)})$ according to (\ref{VBu}). Given  $\gamma^{(i)}_t(k)$, the parameters
$\kappa^{(i)}_{k}$, $\mu^{(i)}_{k}$, $\nu^{(i)}_{k}$ and $(\tau^2_{k})^{(i)}$ can be calculated by (see, e.g., \cite{titterington})
\begin{align*}
\kappa^{(i)}_k&=\kappa_0+g^{(i)}_k,\quad
\nu^{(i)}_{k}=\nu_0+g^{(i)}_k,\quad g^{(i)}_k=\sum_{t=1}^n \gamma^{(i)}_t(k),\\
\mu^{(i)}_{k}&={\kappa_0\over \kappa_0+g^{(i)}_k }\xi_k+{g^{(i)}_k\over \kappa_0+g^{(i)}_k}
\tilde{x}^{(i)}_k,\quad \tilde{x}^{(i)}_k={1\over g^{(i)}_k}\sum_{t=1}^n \gamma^{(i)}_t(k) x_t,\\
\nu^{(i)}_k (\tau^2_{k})^{(i)}&=\nu_0\tau_0^2+\sum_{t=1}^n(x_t-\tilde{x}^{(i)}_k)^2\gamma^{(i)}_t(k)+{\kappa_0
g^{(i)}_k\over \kappa_0+g^{(i)}_k}(\tilde{x}^{(i)}_k-\xi_k)^2.\end{align*}
Compute then $h^{(i+1)}_k(x_t)$ as in (\ref{hht}). With help of
$h^{(i+1)}_k(x_t)$ and $u^{(i+1)}_{lj}$,  find $\gamma^{(i+1)}_t$
and $\xi^{(i+1)}(l,j)$ using the forward-backward formulae. After
that update $h^{(i+2)}_k(x_t)$ and $u^{(i+2)}_{lj}$ and so on.
When the VB algorithm has converged, say after $m$ steps,
apply the Viterbi algorithm with $u_{lj}^{(m)}$ as transitions and
with $h^{(m)}_k(x_t)$ as emission values.
%
\subsection{Clustering formulae under normal emissions with NIX priors}
From (\ref{jaana}) it follows that for any sequence $y'$, the likelihood ratio is given by
\begin{equation}\label{ratio}
{p(x|y)\over p(x|y')}=\prod_{k=1}^K{\Gamma({\nu_0+m_k\over 2})\over
\Gamma({\nu_0+m'_k\over 2})} \prod_{k=1}^K{\sqrt{\kappa_0+m'_k\over
\kappa_0+m_k}} \prod_{k=1}^K {(\nu'_{k}\tau^{'2}_{k})^{{\nu_{0}\over
2}}\over (\nu_{k}\tau^2_{k})^{{\nu_0\over
2}}} \prod_{k=1}^K {(\nu'_{k}\tau^{'2}_{k})^{{m'_k\over
2}}\over (\nu_{k}\tau^2_{k})^{{m_k\over
2}}}.\end{equation}
When $\nu_0\to \infty$ and $\tau^2_0>0$, then due to $\sum_k m_k=\sum_k m'_k=n$ we have
$$\lim_{\nu_0\to \infty}\prod_{k=1}^K{\Gamma({\nu_0+m_k\over 2})\over
\Gamma({\nu_0+m'_k\over 2})}=1,\quad \lim_{\nu_0\to \infty} \prod_{k=1}^K
{(\nu'_{k}\tau^{'2}_{k})^{{m'_k\over 2}}\over
(\nu_{k}\tau^2_{k})^{{m_k\over 2}}}= 1.$$
Write $\nu_{k}\tau^{2}_{k}$ as
$$\nu_{k}\tau^{2}_{k}=\nu_0\tau_0^2 + \sum_{t\in S_k} (x_t-\bar{x}_k)^2  + \frac{\kappa_0 m_k
}{\kappa_0 +m_k}(\bar{x}_k-\xi_k)^2=\nu_0\tau_0^2+A_k=\nu_0\tau_0^2\left(1+{2A_k\over 2\nu_0\tau_0^2}\right).$$
Then
$$ {(\nu'_{k}\tau^{'2}_{k})^{{\nu_{0}\over 2}}\over
(\nu_{k}\tau^2_{k})^{{\nu_0\over 2}}}={\Big(1+{2A'_k\over
2\nu_0\tau_0^2}\Big)^{\nu_0\over 2} \over \Big(1+{2A_k\over
2\nu_0\tau_0^2}\Big)^{\nu_0\over 2}}\to \exp\left[{A'_k-A_k\over
2\tau^2_0}\right].$$
Therefore, when $\nu_0\to \infty$, then the likelihood ratio in (\ref{ratio}) converges to
$$\prod_{k=1}^K{\sqrt{\kappa_0+m'_k\over
\kappa_0+m_k}} \exp\left[{\sum_k A'_k-\sum_kA_k\over 2\tau^2_0}\right].$$
Thus, maximizing $p(x|y)$ corresponds to the following clustering problem:
find clusters $S_1,\ldots, S_k$ that minimize
\[ \sum_{k=1}^K \sum_{t\in S_k} (x_t-\bar{x}_k)^2+\kappa_0\sum_{k=1}^K{m_k\over \kappa_0+m_k}(\bar{x}_k-\xi_k)^2+\tau_0^2\sum_{k=1}^K\ln (\kappa_0+m_k),
\]
which is formula (\ref{clust}). Given cluster $S_k$, it is easy to see that
\begin{equation}\label{center}
\arg\min_{\mu \in {\mathbb R}}\Big[\sum_{t\in S_k}(x_t-\mu)^2+\kappa_0(\mu-\xi_k)^2\Big]={m_k{\bar
x}_k+\kappa_0\xi_k\over \kappa_0+m_k}=:\mu_k.\end{equation}
Since
\[ \sum_{t\in S_k}(x_t-\mu_k)^2+\kappa_0(\mu_k-\xi_k)^2=\sum_{t\in S_k}
(x_t-\bar{x}_k)^2+\kappa_0{m_k\over \kappa_0+m_k}(\bar{x}_k-\xi_k)^2,\]
we obtain (\ref{clustequiv}).

To understand the behavior of the segmentation EM and segmentation MM
algorithms when $\nu_0\to \infty$, recall the segmentation EM iteration
formula from (\ref{hht}).
When $\nu_0\to \infty$, then
$\ln(\nu^{(i)}_{k}/{2})-\psi({\nu^{(i)}_{k}}/{2})\to 0$ and $(\tau_{k}^2)^{(i)}\to \tau_0^2$.
Thus, leaving the superscript $(i)$ out of the notation, we get
\[\ln h_k(x_t) \to -{1\over 2}\ln\big(2\pi
(\tau_0^2)\big)-{1\over 2 \tau^2_0}\big(x_t- \mu_k \big)^2 -{1\over 2(\kappa_0+m_k)}, \]
where $\mu_k$  is as in (\ref{center}).
The Viterbi alignment is now obtained as
$$y_t=\arg\min_{k=1,\ldots,K}\Big[\big(x_t- \mu_k\big)^2+{\tau_0^2 \over
m_k+\kappa_0}\Big].$$
%
%
%
\subsection*{Acknowledgments} This work is supported by  Estonian institutional research funding IUT34-5.
\bibliographystyle{plain}

\bibliography{bayesbib}

\end{document}